\documentclass[final,11pt,onecolumn]{IEEEtran}

\usepackage{graphicx, color}
\usepackage{amssymb}
\usepackage{amstext}
\usepackage{epic,eepic}

\usepackage{latexsym}
\usepackage{times}
\usepackage{amssymb}
\usepackage{amsmath}
\usepackage{epsfig}
\usepackage{cite}
\usepackage{verbatim}
\usepackage{enumerate}
\usepackage{subfigure}
\usepackage{multicol}
\usepackage{pstricks}
\usepackage{pst-node}
\usepackage{setspace}
\usepackage{picinpar}
\usepackage{amsmath, amsfonts}
\usepackage{algorithm, algorithmic}
\newtheorem{theorem}{Theorem}[section]
\newtheorem{lemma}[theorem]{Lemma}
\newtheorem{corollary}[theorem]{Corollary}

\newcommand{\define}{\stackrel{\triangle}{=}}
\newcommand{\xSNR}{\textrm{SNR}}
\newcommand{\xINR}{\textrm{INR}}
\newcommand{\xV}{\mathbf{V}}
\newcommand{\xS}{\mathbf{S}}

\begin{document}
\bibliographystyle{ieeetr}

\title{On the Capacity and Generalized Degrees of Freedom of the $X$ Channel}
\author{\authorblockN{Chiachi Huang, Viveck R. Cadambe, and Syed A. Jafar}\\
\authorblockA{Electrical Engineering and Computer Science\\
University of California Irvine\\
Irvine, California, USA\\
Email: \{chiachih, vcadambe, syed\}@uci.edu}}

\maketitle
\vspace{10pt}
\begin{abstract}
We explore the capacity and generalized degrees of freedom of the two-user Gaussian $X$ channel, i.e. a generalization of the $2$ user interference channel where there is an independent message from each transmitter to each receiver. There are three main results in this paper. First, we characterize the sum capacity of the deterministic X channel model under a symmetric setting. Second, we characterize the generalized degrees of freedom of the Gaussian $X$ channel under a similar symmetric model. Third, we extend the noisy interference capacity characterization previously obtained for the interference channel to the $X$ channel. Specifically, we show that the $X$ channel associated with noisy (very weak) interference channel has the same sum capacity as the noisy interference channel.
\end{abstract}
\newpage
\section{Introduction}
Recent research in multi-user information theory has been characterized by a surge of interest in the study of capacity regions of wireless Gaussian networks. Much of this interest has been fueled by significant recent progress in the search of the capacity region of wireless interference networks, a classical problem of multi-user information theory. In their seminal work \cite{Etkin_Tse_Wang}, Etkin, Tse and Wang approximated the capacity region of the two-user Gaussian interference channel to within one bit. Further insight into the capacity of the two-user Gaussian interference network was revealed in \cite{MK_int,Shang_Kramer_Chen,Sreekanth_Veeravalli}. These references found that the decoding strategy of treating interference as noise at each receiver in the interference network is capacity optimal for a class of interference channels, known as the ``noisy interference'' channels.  Recent results have also found approximations to the capacity regions of certain classes of the  $K$-user interference channel in the high signal-to-noise ratio (SNR) regime. Reference \cite{Cadambe_Jafar_int} approximated the capacity region of the fully connected $K$-user interference channel with time-varying channel coefficients as $$C(\xSNR) = \frac{K}{2} \log(\xSNR)+o(\log(\xSNR))$$ where $\xSNR$ represents the total transmit power of all nodes when the local noise power at each receiver is normalized to unity. In other words, it was shown that the time-varying $K$-user interference channel has $\frac{K}{2}$ degrees of freedom. Similar capacity approximations of the $K$-user  ($K>2$) interference channel with constant channel coefficients (i.e., not time-varying or frequency-selective) are not known in general.

From the recent advances in the study of interference channels, many interesting and powerful tools related to the study of general wireless networks have emerged. Reference \cite{Etkin_Tse_Wang} introduced the notion of \emph{generalized degrees of freedom} to study the performance of various interference management schemes in the interference channel. As its name suggests, the idea of generalized degrees of freedom is a generalization of the concept of degrees of freedom originally introduced in \cite{Zheng_Tse_Diversity}. The idea of generalized degrees of freedom is powerful because in the multiple access, broadcast and two-user interference channels, achievable schemes that are optimal from a generalized degrees of freedom perspective also achieve within a constant number of bits of capacity \cite{Bresler_Tse}. A useful technique in the characterization of the generalized degrees of freedom of a wireless network is the deterministic approach, originally introduced in the context of relay networks \cite{Avestimehr_Diggavi_Tse}. The deterministic approach essentially maps a Gaussian network to a deterministic channel, i.e, a channel whose outputs are deterministic functions of its inputs. The deterministic channel captures the essential structure of the Gaussian channel, but is significantly simpler to analyze. Reference \cite{Bresler_Tse} showed that the deterministic approach leads to a characterization of the generalized degrees of freedom of wireless networks in the two-user interference network, which leads to a constant bit approximation of its capacity. 

In this paper, we explore the two-user $X$ channel - a network with two transmitters, two receivers and four independent messages - one corresponding to each transmitter-receiver pair. The degrees of freedom of the Gaussian $X$ channel have been found in \cite{MMK,Jafar_Shamai}. This work pursues a more refined characterization in terms of the \emph{generalized} degrees of freedom. Unlike the conventional degrees of freedom perspective where all signals are approximately equally strong in the $dB$ scale, the generalized degrees of freedom perspective provides a richer characterization by allowing the full range of relative signal strengths in the $dB$ scale. For example, consider the interference channel. The strong and weak interference scenarios are not visible in the conventional degrees of freedom perspective but become immediately obvious in the generalized degrees of freedom framework. Now consider the $X$ channel which is a generalization of the interference channel to a scenario where every transmitter has a message to every receiver. One of the key features of the $X$ channel is that, unlike the two-user interference channel, it provides the possibility of \emph{interference alignment} \cite{MMK, Jafar_Shamai}. Interference alignment refers to the construction of signals such that they overlap at receivers where they cause interference, but remain distinguishable at receivers where they are desired. Interference alignment is the key to the degrees of freedom characterizations of the $X$ channel with $2$ or more users \cite{Cadambe_Jafar_X}, and for the interference channel with $3$ or more users \cite{Cadambe_Jafar_int}. Since the potential for interference alignment does not arise in the $2$ user interference channel, the two-user $X$ channel provides the simplest possible setting for interference alignment, in terms of the number of transmitters/receivers and channel coefficients. It is shown in \cite{Jafar_Shamai} that, due to interference alignment, the $2$ user $X$ channel has $4/3$ degrees of freedom (assuming time-varying channels), while the $2$ user interference channel has only $1$ degree of freedom. In this work, we explore this capacity advantage of the $X$ channel over the interference channel in the richer context of the generalized degrees of freedom. Specifically we quantify the benefits of interference alignment in terms of generalized degrees of freedom and identify operating regimes where alignment helps the $X$ channel outperform the interference channel. For simplicity, we will keep the number of channel parameters to a minimum by using the symmetric interference channel as our benchmark and presenting our main results for the corresponding symmetric $X$ channel.

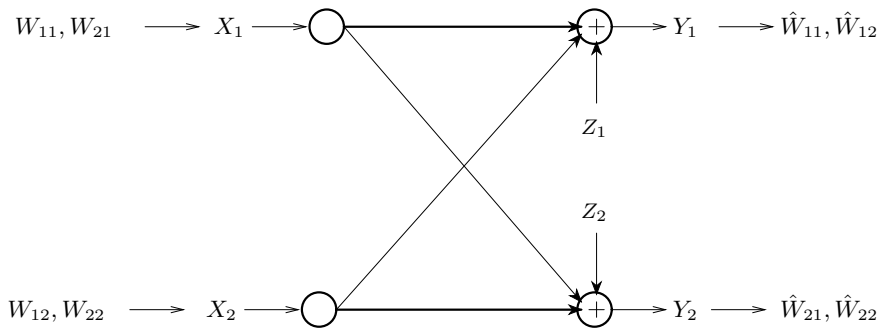
\begin{figure}
\begin{center}\setlength{\unitlength}{0.00054167in}
\begingroup\makeatletter\ifx\SetFigFont\undefined%
\gdef\SetFigFont#1#2#3#4#5{%
  \reset@font\fontsize{#1}{#2pt}%
  \fontfamily{#3}\fontseries{#4}\fontshape{#5}%
  \selectfont}%
\fi\endgroup%
{\renewcommand{\dashlinestretch}{30}
\begin{picture}(10421,3170)(0,-10)
\path(1200,190)(1725,190)
\path(1605.000,160.000)(1725.000,190.000)(1605.000,220.000)
\path(1350,2965)(1875,2965)
\path(1755.000,2935.000)(1875.000,2965.000)(1755.000,2995.000)
\thicklines
\path(3300,2965)(5625,2965)
\blacken\thinlines
\path(5505.000,2920.000)(5625.000,2965.000)(5505.000,3010.000)(5541.000,2965.000)(5505.000,2920.000)
\path(6900,190)(7425,190)
\path(7305.000,160.000)(7425.000,190.000)(7305.000,220.000)
\path(6975,2965)(7500,2965)
\path(7380.000,2935.000)(7500.000,2965.000)(7380.000,2995.000)
\put(6525,115){\makebox(0,0)[lb]{\smash{{{\SetFigFont{9}{10.8}{\rmdefault}{\mddefault}{\updefault}$Y_2$}}}}}
\put(6525,2890){\makebox(0,0)[lb]{\smash{{{\SetFigFont{9}{10.8}{\rmdefault}{\mddefault}{\updefault}$Y_1$}}}}}
\put(5625,1090){\makebox(0,0)[lb]{\smash{{{\SetFigFont{8}{9.6}{\rmdefault}{\mddefault}{\updefault}$Z_2$}}}}}
\put(5625,1915){\makebox(0,0)[lb]{\smash{{{\SetFigFont{8}{9.6}{\rmdefault}{\mddefault}{\updefault}$Z_1$}}}}}
\put(1950,115){\makebox(0,0)[lb]{\smash{{{\SetFigFont{9}{10.8}{\rmdefault}{\mddefault}{\updefault}$X_2$}}}}}
\put(0,115){\makebox(0,0)[lb]{\smash{{{\SetFigFont{9}{10.8}{\rmdefault}{\mddefault}{\updefault}$W_{12},W_{22}$}}}}}
\put(75,2890){\makebox(0,0)[lb]{\smash{{{\SetFigFont{9}{10.8}{\rmdefault}{\mddefault}{\updefault}$W_{11},W_{21}$}}}}}
\put(2025,2890){\makebox(0,0)[lb]{\smash{{{\SetFigFont{9}{10.8}{\rmdefault}{\mddefault}{\updefault}$X_1$}}}}}
\put(7575,115){\makebox(0,0)[lb]{\smash{{{\SetFigFont{9}{10.8}{\rmdefault}{\mddefault}{\updefault}$\hat{W}_{21},\hat{W}_{22}$}}}}}
\put(7575,2890){\makebox(0,0)[lb]{\smash{{{\SetFigFont{9}{10.8}{\rmdefault}{\mddefault}{\updefault}$\hat{W}_{11},\hat{W}_{12}$}}}}}
\path(2325,190)(2850,190)
\path(2730.000,160.000)(2850.000,190.000)(2730.000,220.000)
\path(5775,3040)(5775,2890)
\path(5700,2965)(5850,2965)
\path(5775,265)(5775,115)
\path(5700,190)(5850,190)
\thicklines
\put(3057,190){\ellipse{336}{336}}
\put(3132,2965){\ellipse{336}{336}}
\put(5757,2965){\ellipse{336}{336}}
\put(5757,190){\ellipse{336}{336}}
\thinlines
\path(5925,190)(6450,190)
\path(6330.000,160.000)(6450.000,190.000)(6330.000,220.000)
\path(3300,2965)(5625,265)
\blacken\path(5512.598,326.569)(5625.000,265.000)(5580.797,385.296)(5570.188,328.653)(5512.598,326.569)
\path(5925,2965)(6450,2965)
\path(6330.000,2935.000)(6450.000,2965.000)(6330.000,2995.000)
\path(3225,190)(5625,2890)
\blacken\path(5578.910,2770.415)(5625.000,2890.000)(5511.643,2830.207)(5569.193,2827.218)(5578.910,2770.415)
\path(5775,2215)(5775,2815)
\blacken\path(5820.000,2695.000)(5775.000,2815.000)(5730.000,2695.000)(5775.000,2731.000)(5820.000,2695.000)
\blacken\path(5730.000,460.000)(5775.000,340.000)(5820.000,460.000)(5775.000,424.000)(5730.000,460.000)
\path(5775,340)(5775,940)
\thicklines
\path(3225,190)(5625,190)
\blacken\thinlines
\path(5505.000,145.000)(5625.000,190.000)(5505.000,235.000)(5541.000,190.000)(5505.000,145.000)
\path(2400,2965)(2925,2965)
\path(2805.000,2935.000)(2925.000,2965.000)(2805.000,2995.000)
\end{picture}
}
\caption{The two-user Gaussian $X$ channel}
\end{center}
\end{figure}

Our approach to solving the generalized degrees of freedom of the $X$ channel follows the deterministic approach of \cite{Jafar_Vishwanath_GDOF}. We first introduce the deterministic $X$ channel, and find a tight outerbound and achievable scheme for the sum capacity of this channel in Section \ref{sec:dtrmnstcxc}. In Section \ref{sec:gaussianxc}, we extend the achievability and outerbound arguments of Section \ref{sec:dtrmnstcxc} to the Gaussian $X$ channel yielding its generalized degrees of freedom. A second result we obtain is a generalization of the results of \cite{MK_int,Shang_Kramer_Chen,Sreekanth_Veeravalli} to find the capacity of the Gaussian $X$ channel for a class of channel coefficients.
We introduce the system model, formally define the notion of generalized degrees of freedom, and present the main results in the next section.

\section{System Model}
\label{sec:model}
\subsection{Deterministic $X$ Channel}
\begin{figure}
\begin{center}
\epsfig{figure=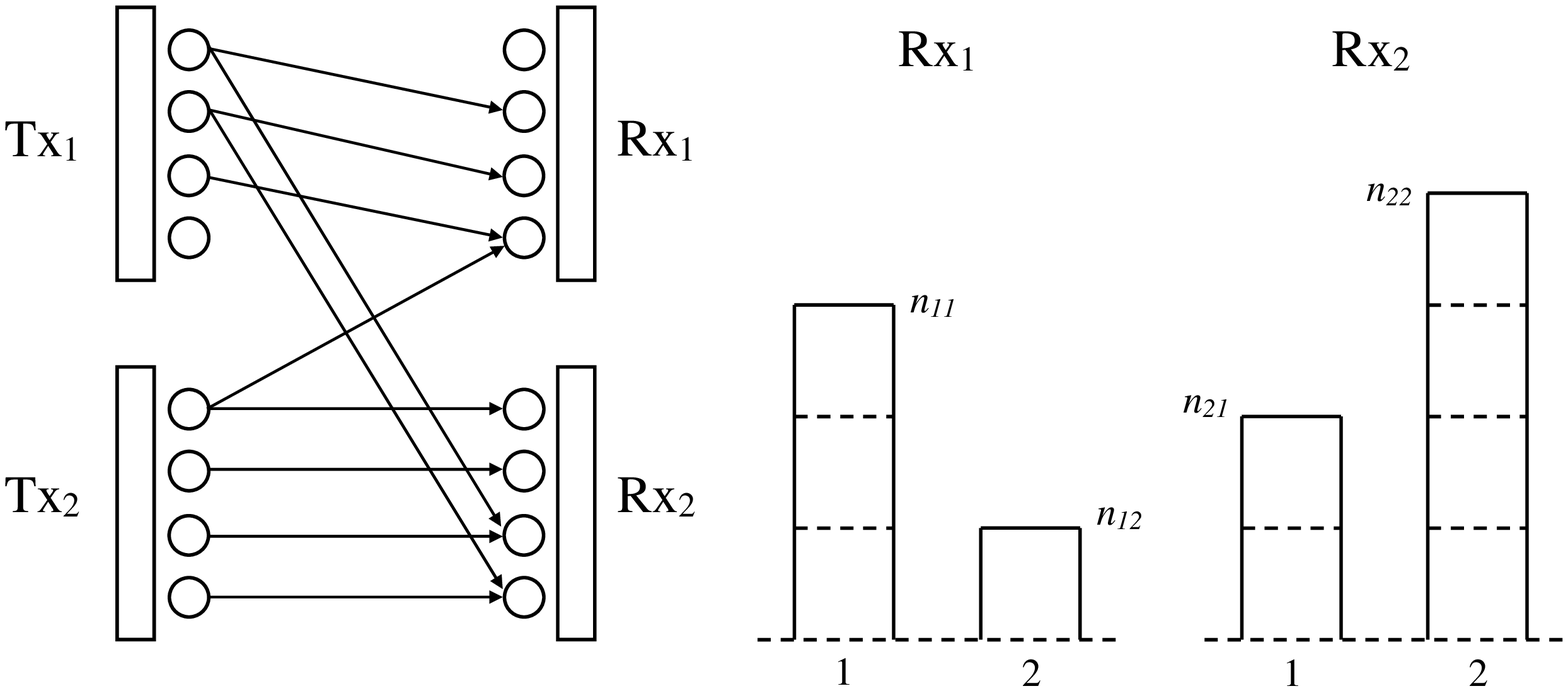, height=3.7in, width=6.1in}
\caption{On the left is an example of the deterministic interference channel. On the right is the figure that shows only the signal levels observed at each receiver.}\label{fig:DtrChnnl}
\end{center}
\end{figure}
The deterministic $X$ channel is physically the same channel as the deterministic interference channel introduced in \cite{Bresler_Tse}, except that the $X$ channel has $4$ independent messages $\{W_{11},W_{12},W_{21},W_{22}\}$ where $W_{ij}$ is the message that originates at transmitter $j$ and is intended for receiver $i$. Note that the interference channel has only $2$ independent messages, e.g.,  $\{W_{11}, W_{22}\}$ or $\{W_{12},W_{21}\}$. The deterministic channel is shown is Fig. \ref{fig:DtrChnnl} and described by the input output equations
\begin{eqnarray}
\mathbf{Y}_{1}(t) = \mathbf{S}^{q-n_{11}} \mathbf{X}_{1}(t) + \mathbf{S}^{q-n_{12}} \mathbf{X}_{2}(t) \\
\mathbf{Y}_{2}(t) = \mathbf{S}^{q-n_{21}} \mathbf{X}_{1}(t) + \mathbf{S}^{q-n_{22}} \mathbf{X}_{2}(t)
\end{eqnarray}
where $q=\max(n_{11}, n_{21}, n_{12}, n_{22})$, $\mathbf{X}_{i}(t), \mathbf{Y}_{i}(t) \in \mathcal{F}_2^q$ for $i=1,2$, and $\mathbf{S}$ is a $q \times q$ shift matrix,
\begin{eqnarray}
\mathbf{S} = \left(
               \begin{array}{ccccc}
                 0      & 0 & 0 & \cdots & 0 \\
                 1      & 0 & 0 & \cdots & 0 \\
                 0      & 1 & 0 & \cdots & 0 \\
                 \vdots &   &   & \ddots & \vdots \\
                 0      & 0 & 0 & \cdots & 1 \\
               \end{array}
             \right)
\end{eqnarray}
The message set and standard definitions and notations of the achievable rates are similar to those in the Gaussian setting. To avoid confusion, sometimes we add the subscript \textit{det} to distinguish the notations for the deterministic channel from those for the Gaussian channel.

\subsection{The Gaussian $X$ Channel}
\label{sec:model_gaussian}
The two-user Gaussian $X$ channel is described by the input-output equations
\begin{eqnarray}
\label{eqn:inoutput1}
Y_{1}(t) &=& H_{11}X_{1}(t) + H_{12} X_{2}(t) + Z_{1}(t) \\
\label{eqn:inoutput2}
Y_{2}(t) &=& H_{21}X_{1}(t) + H_{22} X_{2}(t) + Z_{2}(t)
\end{eqnarray}
where at symbol index $t$, $Y_{j}(t)$ and $Z_{j}(t)$ are the channel output symbol and additive white Gaussian noise (AWGN) respectively at receiver $j$. $X_{i}(t)$ is the channel input symbol at transmitter $i$, and $H_{ji}$ is the channel gain coefficient between transmitter $i$ and receiver $j$ for all $i,j \in \{1,2\}$. All symbols are real and the channel coefficients do not vary w.r.t symbol index. In the remainder of this paper, we suppress time index $t$ if no confusion would be caused. The AWGN is normalized to have zero mean and unit variance and the input power constraint is given by
\begin{eqnarray}
\mathbf{E} \left[ X_{i}^2\right] \leq P_i, & i = 1,2.
\end{eqnarray}

There are four independent messages in the $X$ channel: $W_{11},W_{12},W_{21},W_{22}$ where $W_{ij}$ represents the message from transmitter $j$ to receiver $i$. We indicate the size of the message by $|W_{ij}|$. For codewords spanning $T$ symbols, rates $R_{ij} = \frac{\log|W_{ij}|}{T}$ are achievable if the probability of error for all messages can be simultaneously made arbitrarily small by choosing an appropriate large $T$. The capacity region $\mathcal{C}$ of the $X$ channel is the set of all achievable rate tuples $\mathbf{R} = (R_{11},R_{12},R_{21},R_{22})$. We indicate the sum capacity of the $X$ channel by $C_{\Sigma}$.

\subsubsection{Generalized Degrees of Freedom (GDOF)}
To motivate our problem formulation, we briefly revisit the framework for the generalized degrees of freedom characterization of the symmetric \emph{interference} channel. The interference channel is defined as:
\begin{eqnarray}
\label{eqn:syminout1int}
Y_{1}(t) &=& \sqrt{\xSNR} X_{1}(t) + \sqrt{\xINR} X_{2}(t) + Z_{1}(t) \\
\label{eqn:syminout2int}
Y_{2}(t) &=& \sqrt{\xINR}X_{1}(t) + \sqrt{\xSNR} X_{2}(t) + Z_{2}(t)
\end{eqnarray}
and with the parameter $\alpha$ defined as follows
\begin{eqnarray}
\alpha \triangleq \frac{\log(\xINR)}{\log(\xSNR)}
\end{eqnarray}
the GDOF metric is defined as \cite{Etkin_Tse_Wang},
\begin{eqnarray}
\label{eqn:gdofDef}
d(\alpha) = \limsup_{\xSNR \rightarrow \infty} \frac{C_\Sigma (\xSNR, \alpha)}{\frac{1}{2}\log(\xSNR)}
\end{eqnarray}
where $C_\Sigma (\xSNR, \alpha)$ is the sum capacity of the interference channel.

Since our main goal is to compare GDOF of the $X$ channel with the interference channel, we use the same symmetric interference channel model described above as the physical channel model for the $X$ channel. There is however, one notational difference. Since the terminology $\xSNR, \xINR$ is not as appropriate for the $X$ channel, we instead use the parameter $\rho$ to substitute for these notions, resulting in the following system model for the $X$ channel GDOF characterization:
\begin{eqnarray}
\label{eqn:syminout1}
Y_{1}(t) &=& \sqrt{\rho} X_{1}(t) + \sqrt{\rho^{\alpha }} X_{2}(t) + Z_{1}(t) \\
\label{eqn:syminout2}
Y_{2}(t) &=& \sqrt{\rho^{\alpha}}X_{1}(t) + \sqrt{\rho} X_{2}(t) + Z_{2}(t)
\end{eqnarray}
In other words, we have set $H_{11}=H_{22} = \sqrt{\rho}$, $H_{12} = H_{21} = \sqrt{\rho^{\alpha}}$, and $P_1 = P_2 = 1$. Note that (\ref{eqn:syminout1}), (\ref{eqn:syminout2}) represent the same physical channel as (\ref{eqn:syminout1int}), (\ref{eqn:syminout2int}). However, as mentioned earlier, unlike the interference channel the $X$ channel has $4$ independent messages - one from each transmitter to each receiver. The GDOF characterization for the $X$ channel is defined as:
\begin{eqnarray}
\label{eqn:gdofDefX}
d(\alpha) = \limsup_{\rho \rightarrow \infty} \frac{C_\Sigma (\rho, \alpha)}{\frac{1}{2}\log(\rho)}
\end{eqnarray}
where $C_\Sigma (\rho, \alpha)$ is the sum capacity of the $X$ channel.

Note that we use $\limsup$ to ensure that $d(\alpha)$ always exits. The half in the denominator is because all signals and channel gains are real.

\section{Main Results}
\subsection{Sum Capacity of the Symmetric Deterministic $X$ Channel}
The first main result of the paper is the characterization of the sum capacity of the symmetric deterministic $X$ channel in the symmetric setting where $n_{11}=n_{22}=n_d$ and $n_{12}=n_{21}=n_c$. This result is given in the following theorem.
\begin{theorem}
\label{thm:dtrmnstc_SymSumRate}
The sum capacity $C_\Sigma(n_c,n_d)$ of the symmetric deterministic $X$ channel, i.e., the deterministic $X$ channel where $n_{11}=n_{22}=n_c$ and $n_{12}=n_{21}=n_d$, is
\begin{eqnarray}
\label{eqn:dtrmnstc_SymSumRate}
C_\Sigma(n_c, n_d) &=& \left\{ \begin{array}{ll}
   2n_d - 2n_c,  & 0 \leq \frac{n_c}{n_d} < \frac{1}{2} \\
   2n_c,         & \frac{1}{2} \leq \frac{n_c}{n_d} < \frac{3}{4} \\
   2(n_d - \frac{1}{3}n_c), & \frac{3}{4} \leq \frac{n_c}{n_d} < 1 \\
   n_d,          & n_c = n_d \\
   2(n_c - \frac{1}{3}n_d), & 1 < \frac{n_c}{n_d} \leq \frac{4}{3} \\
   2n_d,         & \frac{4}{3} < \frac{n_c}{n_d} \leq 2 \\
   2n_c - 2n_d,  & \frac{n_c}{n_d} > 2
\end{array}
\right.
\end{eqnarray}
\end{theorem}
\subsection{Generalized Degrees of Freedom of the Symmetric Gaussian $X$ Channel}
The second main result of this paper builds upon the result of Theorem \ref{thm:dtrmnstc_SymSumRate} to find the generalized degrees of freedom characterization (shown in Figure \ref{fig:gdofxc}) for the Gaussian $X$ channel.
\begin{theorem}
\label{thm:gdof}
The generalized degrees of freedom $d(\alpha)$ of the symmetric Gaussian $X$ channel can be characterized as 
\begin{eqnarray}
\label{eqn:gdof}
d(\alpha)= \left\{ \begin{array}{ll}
2-2\alpha,  & 0 \leq \alpha < \frac{1}{2} \\
2\alpha,    & \frac{1}{2} \leq \alpha < \frac{3}{4} \\
2-\frac{2}{3}\alpha   & \frac{3}{4} \leq \alpha < 1 \\
1           & \alpha = 1 \\
2\alpha - \frac{2}{3} & 1 < \alpha \leq \frac{4}{3} \\
2           & \frac{4}{3} < \alpha \leq 2 \\
2\alpha - 2 & \alpha > 2 \end{array} \right.
\end{eqnarray}
\end{theorem}

\begin{figure}
\begin{center}
\includegraphics[angle=90, height=5.3in, width=7.8 in, trim=150 0 0 0, clip]{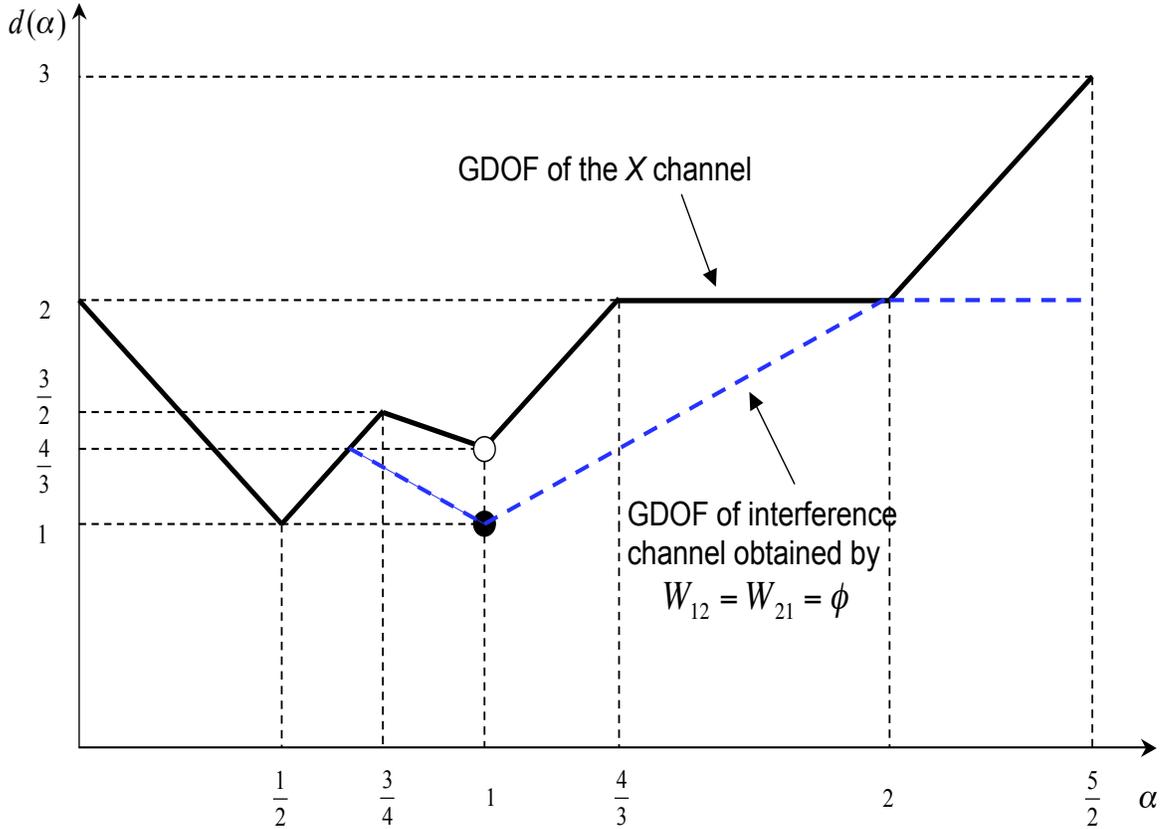}
\caption{Generalized Degrees of Freedom of the symmetric $X$ channel, and a comparison the the $2$ user interference channel}
\label{fig:gdofxc}
\end{center}
\end{figure}

For comparison, Figure \ref{fig:gdofxc} also shows the generalized degrees of freedom characterization of the symmetric interference channel as obtained in \cite{Etkin_Tse_Wang}. For values of $\alpha < 2/3$, characterization of $d(\alpha)$ is identical for both the symmetric two-user Gaussian $X$ channel and the symmetric two-user Gaussian interference channel (See \cite{Etkin_Tse_Wang}  Figure \ref{lem:extdecomp}). We prove this by showing that the Etkin-Tse-Wang (ETW) outerbound derived for the interference channel \cite{Etkin_Tse_Wang} holds for the $X$ channel as well (See Theorem \ref{thm:etw_X}). The ETW outerbound is tight from a GDOF perspective in the interference channel for $\alpha \leq 2/3$. Therefore, our extension of this outerbound implies that for $\alpha \leq 2/3$ a GDOF optimal achievable scheme is to set $W_{12}=W_{12}=\phi$, so that the $X$ channel operates as an interference channel. For example, if $\alpha \leq 1/2$, setting $W_{21}=W_{12}=\phi$ and treating interference as noise is GDOF optimal in the $X$ channel, since it is optimal in the corresponding interference channel \cite{Etkin_Tse_Wang}. Similarly, we show that for $\alpha > \frac{3}{2}$, it is GDOF optimal to set $W_{22}=W_{11}=\phi$ and operate the $X$ channel as an interference channel with messages $W_{12}$ and $W_{21}$. It must be noted that for both $\alpha \leq 2/3$ and $\alpha > 3/2$ the GDOF optimal achievable scheme operates the $X$ channel as \emph{weak} interference channel by setting the appropriate messages to null. For $2/3 < \alpha \leq 3/2$, we propose an interference alignment based achievable scheme for the $X$ channel. Thus, in this regime, the $X$ channel performs better than the interference channel by exploiting the possibility of interference alignment.

\subsection{Capacity of the ``Noisy'' Gaussian $X$ Channel}
References \cite{MK_int,Shang_Kramer_Chen,Sreekanth_Veeravalli} showed that in the interference channel, for a class of channel coefficients, encoding messages using Gaussian codebooks and decoding desired messages by treating interference as noise at each receiver is capacity optimal. Our last main result extends this conclusion to the $X$ channel as well. We show that if a $2$ user interference channel satisfies  the noisy interference conditions obtained in \cite{MK_int,Shang_Kramer_Chen,Sreekanth_Veeravalli} then the corresponding $X$ channel obtained by allowing all transmitters to communicate with all receivers, has the same sum capacity as the original noisy interference channel. This is a surprising result since it implies that for a class of $X$ channels, interference alignment has no capacity benefit. The result holds for the general (asymmetric) $X$ channel and is stated as such in Theorem \ref{thm:AsymNoisyX} in Section \ref{sec:noisyX}. For simplicity we re-state the result here for the symmetric case ($H_{11}=H_{22}=1, H_{12}=H_{21}=h,P_1=P_2=P)$ in a notation consistent with \cite{Sreekanth_Veeravalli}, as follows.

{\it Noisy ``Symmetric'' X Channel Result: If $\left| h \left( 1 + h^2 P \right) \right|  \leq \frac{1}{2},$ then the sum capacity of the Gaussian $X$ channel is given by
$C_{\Sigma} = \log \left( 1 + \frac{ P}{1 + h^2 P} \right)$.
Similarly, if $|h|\geq 2(1 + P)$
then the sum capacity of the Gaussian $X$ channel is given by $C_{\Sigma} =\log \left( 1 + \frac{h^2 P}{1 +  P} \right)$.}

The condition $\left| h \left( 1 + h^2 P \right) \right|  \leq \frac{1}{2}$ is the same as the noisy interference condition in \cite{Sreekanth_Veeravalli}. It means that when the cross-links are too weak, there is no sum-capacity benefit in communicating messages over those links ($X$ channel operation), even though it rules out interference alignment, and we are better off just communicating on the direct links while treating the weak interference as noise. Thus, in this case messages $W_{12},W_{21}$ do not increase sum capacity of the $X$ channel.

The other condition $|h|\geq 2(1 + P)$ refers to a strong cross-channel scenario. It says that when the cross-links are too strong relative to direct links, then sum capacity is achieved by communicating only over the strong cross-links and treating the weak interference received over the direct links as noise. In this case, messages $W_{11},W_{22}$ do not increase the sum capacity of the $X$ channel.

\noindent{\it Notation:} In the rest of this paper, we use the notation 
$$ A^{(T)} \define \left( A(1),A(2), \ldots A(T)\right)$$ 
for any sequence $A$.

\section{Sum Capacity of the Symmetric Deterministic $X$ Channel}
\label{sec:dtrmnstcxc}

The deterministic $X$ channel model is described in the symmetric setting by:
\begin{eqnarray}
\mathbf{Y}_{1}(t) = \mathbf{S}^{q-n_d} \mathbf{X}_{1}(t) + \mathbf{S}^{q-n_c} \mathbf{X}_{2}(t) \\
\mathbf{Y}_{2}(t) = \mathbf{S}^{q-n_c} \mathbf{X}_{1}(t) + \mathbf{S}^{q-n_d} \mathbf{X}_{2}(t)
\end{eqnarray}
where $q = \max(n_c, n_d)$. 

To prove Theorem \ref{thm:dtrmnstc_SymSumRate}, we use the following lemma.
\begin{lemma}
\label{lemma:dtrmnstc_SumRate}
\begin{eqnarray}
\label{eqn:dtrmnstc_SymSwtch}
C_{\Sigma}(n_c, n_d) =
C_{\Sigma}(n_d, n_c) 
\end{eqnarray}
\end{lemma}
The lemma follows trivially from the symmetry in the $X$ channel.
We now proceed to derive the converse argument for Theorem \ref{thm:dtrmnstc_SymSumRate}.


\subsection{Upperbounds}
\label{sec:dtrmnstcxc_up}
In this section, we start from the capacity outerbounds for the (asymmetric) deterministic $X$ channel, and then we use the results to derive the capacity outerbounds for the symmetric setting. The following lemma provides a set of outerbounds for the achievable rate tuple $(R_{11}, R_{21}, R_{12}, R_{22})$ of the (asymmetric) deterministic $X$ channel.
\begin{theorem}
\label{lemma:dtrmnstcxc_up}
The achievable rate tuple $(R_{11}, R_{21}, R_{12}, R_{22})$ of the deterministic $X$ channel satisfies the following inequalities.
\begin{eqnarray}
& &
\label{eqn:dtrmnstc_up1}
R_{11} + R_{12} + R_{22} \leq
\max \left( n_{11}, n_{12} \right) + (n_{22} - n_{12})^+ \\
& &
\label{eqn:dtrmnstc_up2}
R_{11} + R_{21} + R_{22} \leq
\max \left( n_{21}, n_{22} \right) + (n_{11} - n_{21})^+ \\
& &
\label{eqn:dtrmnstc_up3}
R_{11} + R_{21} + R_{12} \leq
\max \left( n_{11}, n_{12} \right) + (n_{21} - n_{11})^+ \\
& &
\label{eqn:dtrmnstc_up4}
R_{21} + R_{12} + R_{22} \leq
\max \left( n_{21}, n_{22} \right) + (n_{12} - n_{22})^+ \\
& &
\label{eqn:dtrmnstc_up5}
R_{11} + R_{21} + R_{12} + R_{22} \leq
\max \left( n_{12}, (n_{11} - n_{21})^+ \right) +
\max \left( n_{21}, (n_{22} - n_{12})^+ \right) \\
& &
\label{eqn:dtrmnstc_up6}
R_{11} + R_{21} + R_{12} + R_{22} \leq
\max \left( n_{11}, (n_{12} - n_{22})^+ \right) +
\max \left( n_{22}, (n_{21} - n_{11})^+ \right)
\end{eqnarray}
\smallskip
\end{theorem}
\begin{proof}


The bound on $R_{11} + R_{12} + R_{22}$, (\ref{eqn:dtrmnstc_up1}), is proved by a genie upperbound. Consider a genie-aided channel where a genie provides $\mathbf{S}^{q-n_{12}} \mathbf{X}_{2}$, $W_{12}$, and $\mathbf{X}_{1}$ to receiver $2$. For a block length $T$, we can bound $R_{22}$ as follows.
\begin{eqnarray}
T (R_{22}-\epsilon)
& \leq &
I \left( W_{22}; {\mathbf{Y}_{2}}^{(T)}, \mathbf{S}^{q-n_{12}}{\mathbf{X}_{2}}^{(T)}, W_{12}, {\mathbf{X}_{1}}^{(T)} \right) \\
& = &
I \left( W_{22}; {\mathbf{X}_{1}}^{(T)} \right) +
I \left( W_{22}; {\mathbf{Y}_{2}}^{(T)}, \mathbf{S}^{q-n_{12}}{\mathbf{X}_{2}}^{(T)}, W_{12} \mid {\mathbf{X}_{1}}^{(T)} \right) \\
& = &
I \left( W_{22}; \mathbf{S}^{q-n_{22}}{\mathbf{X}_{2}}^{(T)}, \mathbf{S}^{q-n_{12}}{\mathbf{X}_{2}}^{(T)}, W_{12} \right) \label{eq:dtob1}\\
& = &
I \left( W_{22}; \mathbf{S}^{q-n_{22}}{\mathbf{X}_{2}}^{(T)}, \mathbf{S}^{q-n_{12}}{\mathbf{X}_{2}}^{(T)} \mid W_{12} \right) \label{eq:dtob2} \\
& = &
I \left( W_{22}; \mathbf{S}^{q-n_{12}}{\mathbf{X}_{2}}^{(T)} \mid W_{12} \right) +
I \left( W_{22}; \mathbf{S}^{q-n_{22}}{\mathbf{X}_{2}}^{(T)} \mid W_{12}, \mathbf{S}^{q-n_{12}}{\mathbf{X}_{2}}^{(T)} \right) \\
& = &
H \left( \mathbf{S}^{q-n_{12}}{\mathbf{X}_{2}}^{(T)} \mid W_{12} \right) -
H \left( \mathbf{S}^{q-n_{12}}{\mathbf{X}_{2}}^{(T)} \mid W_{12}, W_{22} \right) \notag \\
& &
+ H \left( \mathbf{S}^{q-n_{22}}{\mathbf{X}_{2}}^{(T)} \mid W_{12}, \mathbf{S}^{q-n_{12}}{\mathbf{X}_{2}}^{(T)} \right) -
H \left( \mathbf{S}^{q-n_{22}}{\mathbf{X}_{2}}^{(T)} \mid W_{12}, \mathbf{S}^{q-n_{12}}{\mathbf{X}_{2}}^{(T)}, W_{22} \right) \\
& = &
\label{eqn:dtrmnstc_up1_PrfTmp1}
H \left( \mathbf{S}^{q-n_{12}}{\mathbf{X}_{2}}^{(T)} \mid W_{12} \right) +
H \left( \mathbf{S}^{q-n_{22}}{\mathbf{X}_{2}}^{(T)} \mid W_{12}, \mathbf{S}^{q-n_{12}}{\mathbf{X}_{2}}^{(T)} \right)
\end{eqnarray}
where (\ref{eq:dtob1}) and (\ref{eq:dtob2}) hold because all messages are independent of each other. (\ref{eqn:dtrmnstc_up1_PrfTmp1}) follows from the fact that $\mathbf{X}_{2}$ is a function of $W_{12},W_{22}$.
Using Fano's inequality, $R_{11} + R_{12}$ can be bounded as follows.
\begin{eqnarray}
T (R_{11} + R_{12} - \epsilon)
& \leq &
I \left( W_{11}, W_{12}; {\mathbf{Y}_{1}}^{(T)} \right) \\
& = &
H \left( {\mathbf{Y}_{1}}^{(T)} \right) -
H \left({\mathbf{Y}_{1}}^{(T)} \mid  W_{11}, W_{12} \right) \\
& \leq &
H \left( {\mathbf{Y}_{1}}^{(T)} \right) -
H \left({\mathbf{Y}_{1}}^{(T)} \mid  W_{11}, W_{12}, W_{21} \right) \\
& = &
H \left( {\mathbf{Y}_{1}}^{(T)} \right) -
H \left({\mathbf{Y}_{1}}^{(T)} \mid  W_{11}, W_{12}, W_{21}, {\mathbf{X}_{1}}^{(T)} \right) \\
& = &
\label{eqn:dtrmnstc_up1_PrfTmp2}
H \left( {\mathbf{Y}_{1}}^{(T)} \right) -
H \left( \mathbf{S}^{q-n_{12}}{\mathbf{X}_{2}}^{(T)} \mid  W_{12} \right)
\end{eqnarray}
Adding (\ref{eqn:dtrmnstc_up1_PrfTmp1}) and (\ref{eqn:dtrmnstc_up1_PrfTmp2}), we have
\begin{eqnarray}
T (R_{11} + R_{12} + R_{22} - \epsilon)
& \leq &
H \left( {\mathbf{Y}_{1}}^{(T)} \right) +
H \left( \mathbf{S}^{q-n_{22}}{\mathbf{X}_{2}}^{(T)} \mid W_{12}, \mathbf{S}^{q-n_{12}}{\mathbf{X}_{2}}^{(T)} \right) \\
& \leq &
T \left( \max \left( n_{11}, n_{12} \right) + (n_{22} - n_{12})^+ \right).
\end{eqnarray}
Letting $T \rightarrow \infty (\epsilon \rightarrow 0)$, we prove (\ref{eqn:dtrmnstc_up1}). Similarly, we can prove (\ref{eqn:dtrmnstc_up2}), (\ref{eqn:dtrmnstc_up3}), and (\ref{eqn:dtrmnstc_up4}).

Next, the first bound on $R_{11} + R_{21} + R_{12} + R_{22}$, (\ref{eqn:dtrmnstc_up5}), can be proved as follows. Consider a genie-aided channel where a genie provides $\mathbf{S}^{q-n_{21}} \mathbf{X}_{1}$ and $W_{21}$ to receiver $1$. For a block length $T$, using Fano's inequality, we can bound $R_{11} + R_{12}$ as the following.
\begin{eqnarray}
T(R_{11} + R_{12} - \epsilon)
& \leq &
I( W_{11}, W_{12}; {\mathbf{Y}_{1}}^{(T)}, \mathbf{S}^{q-n_{21}} {\mathbf{X}_{1}}^{(T)}, W_{21} ) \\
& = &
I( W_{11}, W_{12}; {\mathbf{Y}_{1}}^{(T)}, \mathbf{S}^{q-n_{21}} {\mathbf{X}_{1}}^{(T)} \mid W_{21} ) \label{eq:dtob3} \\
& = &
I( W_{11}, W_{12}; \mathbf{S}^{q-n_{21}} {\mathbf{X}_{1}}^{(T)} \mid W_{21} ) +
I( W_{11}, W_{12}; {\mathbf{Y}_{1}}^{(T)} \mid \mathbf{S}^{q-n_{21}} {\mathbf{X}_{1}}^{(T)}, W_{21} ) \\
& = &
H( \mathbf{S}^{q-n_{21}} {\mathbf{X}_{1}}^{(T)} \mid W_{21} ) -
H( \mathbf{S}^{q-n_{21}} {\mathbf{X}_{1}}^{(T)} \mid W_{21}, W_{11}, W_{12}) \notag \\
& &
+ H( {\mathbf{Y}_{1}}^{(T)} \mid \mathbf{S}^{q-n_{21}} {\mathbf{X}_{1}}^{(T)}, W_{21} ) -
H( {\mathbf{Y}_{1}}^{(T)} \mid \mathbf{S}^{q-n_{21}} {\mathbf{X}_{1}}^{(T)}, W_{21}, W_{11}, W_{12} ) \\
& = &
H( \mathbf{S}^{q-n_{21}} {\mathbf{X}_{1}}^{(T)} \mid W_{21} ) +
H( {\mathbf{Y}_{1}}^{(T)} \mid \mathbf{S}^{q-n_{21}} {\mathbf{X}_{1}}^{(T)}, W_{21} ) \notag \\
& &
- H( {\mathbf{Y}_{1}}^{(T)} \mid \mathbf{S}^{q-n_{21}} {\mathbf{X}_{1}}^{(T)}, W_{21}, W_{11}, W_{12}, {\mathbf{X}_{1}}^{(T)} ) \\
& = &
H( \mathbf{S}^{q-n_{21}} {\mathbf{X}_{1}}^{(T)} \mid W_{21} ) +
H( {\mathbf{Y}_{1}}^{(T)} \mid \mathbf{S}^{q-n_{21}} {\mathbf{X}_{1}}^{(T)}, W_{21} ) \notag \\
& &
- H( {\mathbf{Y}_{1}}^{(T)} \mid W_{12}, {\mathbf{X}_{1}}^{(T)} ) \\
& = &
\label{eqn:dtrmnstc_up5_PrfTmp1}
H( \mathbf{S}^{q-n_{21}} {\mathbf{X}_{1}}^{(T)} \mid W_{21} ) +
H( {\mathbf{Y}_{1}}^{(T)} \mid \mathbf{S}^{q-n_{21}} {\mathbf{X}_{1}}^{(T)}, W_{21} )
\notag \\
& &
- H( \mathbf{S}^{q-n_{12}} {\mathbf{X}_{2}}^{(T)} \mid W_{12})
\end{eqnarray}
Similarly, we have
\begin{eqnarray}
\label{eqn:dtrmnstc_up5_PrfTmp2}
T(R_{21} + R_{22}-\epsilon) \leq
H( \mathbf{S}^{q-n_{12}} {\mathbf{X}_{2}}^{(T)} \mid W_{12} ) +
H( {\mathbf{Y}_{2}}^{(T)} \mid \mathbf{S}^{q-n_{12}} {\mathbf{X}_{2}}^{(T)}, W_{12} ) -
H( \mathbf{S}^{q-n_{21}} {\mathbf{X}_{1}}^{(T)} \mid W_{21})
\end{eqnarray}
Adding (\ref{eqn:dtrmnstc_up5_PrfTmp1}) and (\ref{eqn:dtrmnstc_up5_PrfTmp2}), we have
\begin{eqnarray}
T(R_\Sigma - \epsilon)
& \leq &
H( {\mathbf{Y}_1}^{(T)} \mid \mathbf{S}^{q-n_{21}} {\mathbf{X}_{1}}^{(T)}, W_{21} ) +
H( {\mathbf{Y}_2}^{(T)} \mid \mathbf{S}^{q-n_{12}} {\mathbf{X}_{2}}^{(T)}, W_{12} ) \\
& \leq &
H( {\mathbf{Y}_1}^{(T)} \mid \mathbf{S}^{q-n_{21}} {\mathbf{X}_{1}}^{(T)} ) + H( {\mathbf{Y}_2}^{(T)}
		\mid \mathbf{S}^{q-n_{12}} {\mathbf{X}_{2}}^{(T)}) \\
& \leq &
T \left( \max ( n_{12}, (n_{11} - n_{21})^+ ) +
\max ( n_{21}, (n_{22} - n_{12})^+ ) \right).
\end{eqnarray}
Letting $T \rightarrow \infty$, we prove (\ref{eqn:dtrmnstc_up5}). Similarly, we can prove (\ref{eqn:dtrmnstc_up6}).
\end{proof}

After obtaining capacity outerbounds for the deterministic $X$ channel, we use them to derive sum-capacity upperbounds for the symmetric case.
\begin{corollary}
\label{cor:dtrmnstc_SymSumRateUp}
For any achievable scheme, the sum-rate $R_{\Sigma}$ can be bounded as 
\begin{eqnarray*}
R_{\Sigma} \leq R_{\Sigma,up} &\define& \left\{ \begin{array}{ll}
   2n_d - 2n_c,  & 0 \leq \frac{n_c}{n_d} < \frac{1}{2} \\
   2n_c,         & \frac{1}{2} \leq \frac{n_c}{n_d} < \frac{3}{4} \\
   2(n_d - \frac{1}{3}n_c), & \frac{3}{4} \leq \frac{n_c}{n_d} < 1 \\
   n_d,          & n_c = n_d \\
   2(n_c - \frac{1}{3}n_d), & 1 < \frac{n_c}{n_d} \leq \frac{4}{3} \\
   2n_d,         & \frac{4}{3} < \frac{n_c}{n_d} \leq 2 \\
   2n_c - 2n_d,  & \frac{n_c}{n_d} > 2
\end{array}
\right.
\end{eqnarray*}
\end{corollary}

\begin{proof}
Consider any reliable coding scheme achieving sum rate $R_{\Sigma}$.
Then we can write
\begin{eqnarray}
& &
\label{eqn:dtrmnstcxc_SymUp1}
R_\Sigma \leq \frac{4}{3} \max \left( n_c, n_d \right) +
\frac{2}{3} (n_d - n_c)^+ + \frac{2}{3} (n_c - n_d)^+ \\
& &
\label{eqn:dtrmnstcxc_SymUp2}
R_\Sigma \leq 2 \max \left( n_c, n_d - n_c \right) \\
& &
\label{eqn:dtrmnstcxc_SymUp3}
R_\Sigma \leq 2 \max \left( n_d, n_c - n_d \right)
\end{eqnarray}
Inequalities (\ref{eqn:dtrmnstcxc_SymUp2}) and (\ref{eqn:dtrmnstcxc_SymUp3}) are direct results of (\ref{eqn:dtrmnstc_up5}) and (\ref{eqn:dtrmnstc_up6}). To prove inequality (\ref{eqn:dtrmnstcxc_SymUp1}), we do the following. Substituting $n_{11} = n_{22} = n_d$ and $n_{12} = n_{21} = n_c$ into (\ref{eqn:dtrmnstc_up1}) to (\ref{eqn:dtrmnstc_up4}), adding the resulting inequalities together, and dividing both sides by $3$, we obtain (\ref{eqn:dtrmnstcxc_SymUp1}).

Further, for the symmetric deterministic $X$ channel, if $n_c = n_d$, then both receivers receive the same signals. Thus, the achievable sum rate is bounded by the multiple access channel bound.
\begin{eqnarray}
\label{eqn:dtrmnstcxc_SymUp4}
R_\Sigma \leq n_d
\end{eqnarray}
The result of Corollary \ref{cor:dtrmnstc_SymSumRateUp} follows from  (\ref{eqn:dtrmnstcxc_SymUp1})-(\ref{eqn:dtrmnstcxc_SymUp4}).
\end{proof}

\subsection{Achievable Schemes}
\label{sec:dtrmnstcxc_low}
The following theorem gives the sum capacity of the symmetric deterministic $X$ channel.
\begin{theorem}
\label{thm:dtrmnstc_SymSumRateAch}
The sum-capacity upperbound given in (\ref{cor:dtrmnstc_SymSumRateUp}) is achievable. Equivalently,
\begin{eqnarray}
C_\Sigma(n_c, n_d) = R_{\Sigma,up}(n_c, n_d)
\end{eqnarray}
\end{theorem}

Before we proceed to the proof, we will need the following lemma

\begin{lemma}
\label{lem:extdecomp}
Let $n_c, n_d$ be positive integers such that $\frac{3}{4} \leq \frac{n_c}{n_d} < 1$. Then
\begin{enumerate}
\item If $n_c$ is divisible by $3$, then there exists a $n_d\times \frac{n_c}{3}$ matrix $\xV$ whose entries are from $\mathcal{F}_2^{n_d}$ such that 
$$ \textrm{rank}\left(\left[ \xV ~~~ \xS^{n_d-n_c} \xV ~~~ \xS^{2n_d-2n_c} \xV ~~~ \xV_{null}\right] \right) = n_d$$
where $\xV_{null}$ is a $n_d\times(n_d-n_c)$ whose column vectors form a basis for the nullspace of $\xS^{n_d-n_c}$
\item There exists a $3n_d\times n_c$ matrix $\mathbf{\bar{V}}$ whose entries are from $\mathcal{F}_2^{3n_d}$ such that 
$$ \textrm{rank}\left(\left[ \mathbf{\bar{V}} ~~~ \mathbf{\bar{H}} \mathbf{\bar{V}}  ~~~ \mathbf{\bar{H}}^2 \mathbf{\bar{V}}  ~~ \mathbf{\bar{V}}_{null} \right] \right) = 3 n_d$$
where 
$$
\mathbf{\bar{H}} = \left[ \begin{array}{ccc}
\mathbf{S}^{n_d-n_c} & \mathbf{0}_{n_d\times n_d}           & \mathbf{0}_{n_d\times n_d} \\
\mathbf{0}_{n_d\times n_d}           & \mathbf{S}^{n_d-n_c} & \mathbf{0}_{n_d\times n_d} \\
\mathbf{0}_{n_d\times n_d}           & \mathbf{0}_{n_d\times n_d}           & \mathbf{S}^{n_d-n_c}
\end{array} \right]
$$
and $\mathbf{\bar{V}}_{null}$ represents the $3n_d \times (3 n_d-3n_c)$ matrix whose column vectors form a basis for the nullspace of $\mathbf{\bar{H}}$
\end{enumerate}
\end{lemma}
The proof of the lemma is placed in Appendix \ref{sec:ProofLemmaDecomp}. We now proceed to prove Theorem \ref{thm:dtrmnstc_SymSumRateAch}.\\
\begin{proof}
We only discuss the achievable scheme for the case that $n_c \leq n_d$. The achievable schemes for $n_c > n_d$ can be obtained by using Corollary \ref{lemma:dtrmnstc_SumRate}. For the case that $n_c \leq n_d$, the achievable scheme is split into four different regimes viz. $0 \leq \frac{n_c}{n_d} < \frac{2}{3}$, $\frac{2}{3} \leq \frac{n_c}{n_d} < \frac{3}{4}$, $\frac{3}{4} \leq \frac{n_c}{n_d} < 1$, and $\frac{n_c}{n_d} = 1$.

Achievability for $\frac{n_c}{n_d}=1$ is trivial, since an optimal achievable scheme sets $W_{12}=W_{21}=W_{22}=\phi$ and uses all the $n_d$ levels for $W_{11}$ at transmitter $1$.
We will treat the other $3$ cases below.

\emph{Case 1 : $0 \leq \frac{n_c}{n_d} < \frac{2}{3} $ } \\
We need to show that $\max(2n_d-2n_c,2n_c)$ is achievable. Achievability follows by setting $W_{21}=W_{12}=\phi$ so that the $X$ channel operates as an interference channel. The capacity of the two-user deterministic interference channel found in \cite{Bresler_Tse,Gamal_Costa} implies that $\max(2n_d-2n_c,2n_c)$ is achievable in this regime.

\emph{Case 2 : $\frac{2}{3} \leq \frac{n_c}{n_d} < \frac{3}{4} $ } \\
We show that a sum rate of $R_{\Sigma}=2n_c$ is achievable in this regime using interference alignment over the deterministic set up. The achievable scheme achieves a rate of $R_{ii}=2n_c-n_d $ for each of $W_{11}, W_{22}$, and a rate of $R_{ij}=n_d-n_c$ for $W_{12}$ and $W_{21}$.

At transmitter $i$, the top $n_d - n_c$ levels are used to transmit $W_{ii}$, the next $n_d - n_c$ levels are used to transmit $W_{ji}$, the next $n_d - n_c$ levels are kept zero, and the remaining $3n_c-2n_d$ levels are used to transmit $W_{ii}$ for $(i,j)=(1,2), (2,1)$ (See Figure \ref{fig:DtrAchv2by3}).
In other words, the achievable scheme transmits for $i \neq j, i,j \in \{1,2\}$, a $n_d \times 1$ column vector $\mathbf{X}_i$ which can be represented as
\begin{eqnarray*} \mathbf{X}_i &=& \left[ \begin{array}{c}\mathbf{I}_{(n_d-n_c)} \\ \mathbf{0}_{n_c \times (n_d-n_c)}\end{array}\right] \mathbf{\hat{X}}_{ii}(1)  + \left[ \begin{array}{c}\mathbf{0}_{(3 n_d-3n_c) \times (3n_c-2n_d)} \\ \mathbf{I}_{(3 n_c-2n_d)}\end{array}\right] \mathbf{\hat{X}}_{ii}(2)  \\&&+ \left[ \begin{array}{c}\mathbf{0}_{( n_d-n_c) \times (n_d-n_c)} \\ \mathbf{I}_{ (n_d-n_c)} \\ \mathbf{0}_{(2n_c-n_d)\times (n_d-n_c)}\end{array}\right] \mathbf{\hat{X}}_{ij} \end{eqnarray*}
where $\mathbf{I}_m$ represents the $m \times m$ identity matrix, $\mathbf{\hat{X}}_{ii}(1), \mathbf{\hat{X}}_{ii}(2), \mathbf{\hat{X}}_{ij}$ are column vectors of sizes $ (n_d-n_c) \times 1, (3n_c-2n_d) \times 1, (n_d-n_c) \times 1$ respectively. $\mathbf{\hat{X}}_{ii}(1),\mathbf{\hat{X}}_{ii}(2)$ are used to encode message $W_{ii}$ and $\mathbf{\hat{X}}_{ij}$ is used to encode $W_{ij}$.
As illustrated in Fig. \ref{fig:DtrAchv2by3}, receiver $i$ can recover its intended messages $W_{ii},W_{ij}$ without interference. Thus, we have $R_\Sigma = 2n_c$. Note that at receiver $i$, interference $\mathbf{\hat{X}}_{ji},\mathbf{\hat{X}}_{jj}$ align at levels $n_d-n_c+1, n_d-n_c+2, \cdots, 2(n_d-n_c)$.

\begin{figure}
\begin{center}
\includegraphics[angle=90, height=3.2in, width=6.3 in, trim=140 0 45 0, clip]{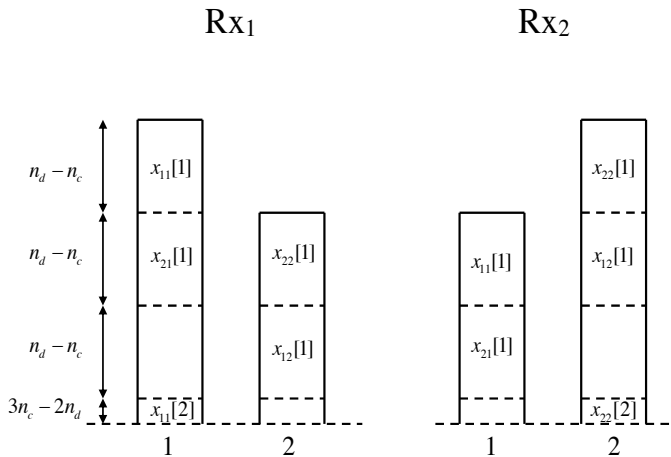}
\caption{Signal levels at receivers for $\frac{2}{3} \leq \frac{n_c}{n_d} < \frac{3}{4}$.}\label{fig:DtrAchv2by3}
\end{center}
\end{figure}

\emph{Case 3 : $\frac{3}{4} \leq \frac{n_c}{n_d} < 1$} \\
We first consider the case where $n_c$ is a multiple of $3$.
For this regime, we show that $R_{\Sigma} = 2n_d-2n_c/3$ is achievable. 

\begin{figure}
\begin{center}
\includegraphics[angle=90, width=400pt, trim=30 0 30 0, clip]{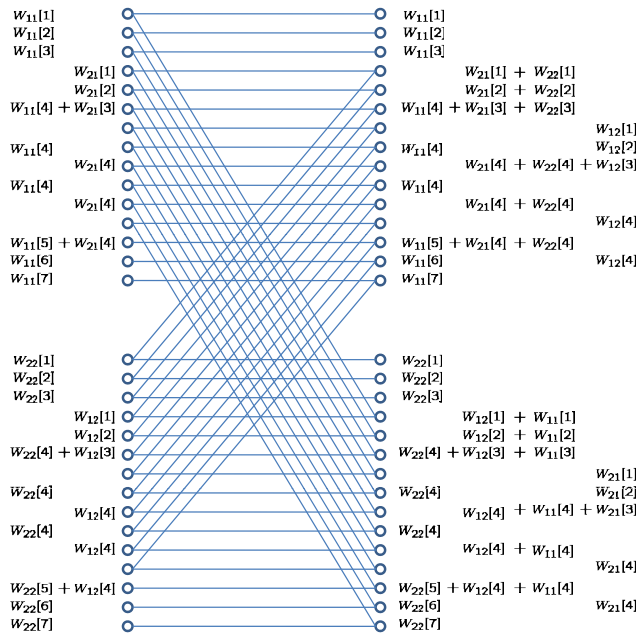}
\caption{Achievable scheme for the symmetric deterministic $X$ channel with $(n_c, n_d) = (12, 15)$}\label{fig:DtrAchv12by15}
\end{center}
\end{figure}

\textit{1) Transmit Scheme:} We use linear precoding at the transmitters. Let $\mathbf{V}_{null}$ be a $n_d \times (n_d-n_c)$ times matrix whose column vectors form a basis for the null space of $\xS^{n_d-n_c}$ meaning that 
$$ \xS^{n_d-n_c} \xV_{null} = \mathbf{0}_{n_d \times (n_d-n_c)}$$ At transmitter $i$, we use, as precoding vectors for $W_{ii}$, column vectors of the  matrix $\left[\xV~~\xV_{null}\right]$ where $\xV$ has dimension $n_d \times \frac{n_c}{3}$. We will shortly explain how $\xV$ is chosen, but here we mention that the columns of $\xV$ are linearly independent of $\xV_{null}$. Note that this implies that $\xS^{n_d-n_c} \xV$ has a full rank of $n_c/3$. For $W_{ji}$, we use $\xS^{n_d-n_c} \xV$ as the precoding matrix so that, the transmitted codeword $\mathbf{X}_i$ can be represented as 

\begin{eqnarray}
\mathbf{X}_{i} =
 \mathbf{V} \mathbf{\hat{X}}_{ii}(1) + \mathbf{V}_{null} \mathbf{\hat{X}}_{ii}(2) + \xS^{n_d-n_c} \xV \mathbf{\hat{X}}_{ji}
\end{eqnarray}
for $(i,j) \in \{(1,2),(2,1)\}$, where $\mathbf{\hat{X}}_{ii}(1)$ and $\mathbf{\hat{X}}_{ii}(2)$ are column vectors of lengths $n_c/3$ and $n_d-n_c$ representing the bits encoding $W_{ii}$. $\mathbf{\hat{X}}_{ji}$ is a $n_c/3 $ dimensional column vector of bits encoding $W_{ji}$.

\textit{2) Receive Scheme:} The received signal at receiver $1$ can be expressed as the following.
\begin{eqnarray}
\mathbf{Y}_{1}
& = &
\mathbf{X}_{1} + \xS^{n_d-n_c} \mathbf{X}_{2} \\
& = &
 \mathbf{V} \mathbf{\hat{X}}_{11}(1) + \mathbf{V}_{null} \mathbf{\hat{X}}_{11}(2) + \xS^{n_d-n_c} \xV \left( \mathbf{\hat{X}}_{21}+\mathbf{\hat{X}}_{22}(1)\right) +  \xS^{2n_d-2n_c} \xV \mathbf{\hat{X}}_{12} 
\end{eqnarray}
Now, receiver $1$ wishes to decode $\mathbf{\hat{X}}_{11}(1),\mathbf{\hat{X}}_{11}(2),\mathbf{\hat{X}}_{12}$ using linear decoding. Notice that the interference from $ \mathbf{\hat{X}}_{21},\mathbf{\hat{X}}_{22}(1)$ aligns along $\xS^{n_d-n_c} \xV$. Now, suppose we choose $\xV$ such that the columns of the matrix 
$$ \mathbf{G} =\left[ \xV ~~~ \xS^{n_d-n_c} \xV ~~~ \xS^{2n_d-2n_c} \xV ~~ \xV_{null}\right] $$
are linearly independent, then clearly  receiver $1$ can decode $W_{11},W_{12}$ using linear decoding. Therefore, in order to show achievability, we need to show that there exists $\xV$ so that the matrix $\mathbf{G}$ has a full rank of $n_d$. This is shown in Lemma \ref{lem:extdecomp}. A similar analysis shows that, if $\mathbf{G}$ has full rank, then receiver $2$ can decode its desired messages as well, using linear decoding.

Now, we consider the case where $n_c/3$ is not an integer. In this case, we use a $3$ symbol extension  of the channel represented below (Channel extensions have earlier been used in achievable schemes in \cite{Jafar_Shamai, Cadambe_Jafar_int, Lapidoth_Shamai_Wigger_IN,Weingarten_Shamai_Kramer}) 
\begin{eqnarray}
\underbrace{\left[ \begin{array}{c}
\mathbf{Y}_{i}(3t) \\
\mathbf{Y}_{i}(3t+1) \\
\mathbf{Y}_{i}(3t+2)
\end{array} \right]}_{\displaystyle{\mathbf{\bar{Y}}_{i}}}
=
\underbrace{\left[ \begin{array}{c}
\mathbf{X}_{i}(3t) \\
\mathbf{X}_{i}(3t+1) \\
\mathbf{X}_{i}(3t+2)
\end{array} \right]}_{\displaystyle{\mathbf{\bar{X}}_{i}}}
+
\underbrace{\left[ \begin{array}{ccc}
\mathbf{S}^{n_d-n_c} & \mathbf{0}           & \mathbf{0} \\
\mathbf{0}           & \mathbf{S}^{n_d-n_c} & \mathbf{0} \\
\mathbf{0}           & \mathbf{0}           & \mathbf{S}^{n_d-n_c}
\end{array} \right]}_{\displaystyle{{\mathbf{\bar{H}}}}}
\underbrace{\left[ \begin{array}{c}
\mathbf{X}_{j}(3t) \\
\mathbf{X}_{j}(3t+1) \\
\mathbf{X}_{j}(3t+2)
\end{array} \right]}_{\displaystyle{\mathbf{\bar{X}}_{j}}}
\end{eqnarray}
Notice that, over this extended channel, inputs and outputs are symbols over $\mathcal{F}_2^{3n_d}$. Like the case where $n_c$ was a multiple of $3$, a linear precoding and decoding technique is applicable over this extended channel. The only difference in this case is that, we need to show that there exists a $3n_d \times n_c$ matrix $\mathbf{\bar{V}}$  such that the matrix 
$$ \mathbf{\bar{G}} =\left[ \mathbf{\bar{V}} ~~~ \mathbf{\bar{H}} \mathbf{\bar{V}}  ~~~ \mathbf{\bar{H}}^2 \mathbf{\bar{V}}  ~~~ \mathbf{\bar{V}}_{null} \right] $$
has a full rank of $3n_d$, where $\mathbf{\bar{V}}_{null} $ represents the $(3n_d-3n_c)$ basis elements of the null space of $\mathbf{\bar{H}}$. This is shown in Lemma \ref{lem:extdecomp} as well. 
This completes the proof of achievability. An example of the scheme for the case that $(n_c,n_d)=(12,15)$ is given in Figure \ref{fig:DtrAchv12by15}.

\end{proof}

\section{Generalized Degrees of Freedom of the Symmetric Gaussian $X$ Channel}
\label{sec:gaussianxc}
The main aim of this section is to prove Theorem \ref{thm:gdof}. 

We first derive a useful property of $d(\alpha)$ - the GDOF of the symmetric two-user $X$ channel
\begin{lemma}
\label{lemma:sym}
\begin{eqnarray}
d(\alpha)=\alpha d(\frac{1}{\alpha})
\end{eqnarray}
where $d(\alpha)$ represents the number of GDOF of the symmetric $X$ channel.
\end{lemma}
\bigskip
\begin{proof}
Please see Appendix \ref{sec:ProofLemmaSym} for the proof.
\end{proof}
The lemma is useful since we can first study $d(\alpha)$ for $\alpha \leq 1$ and then use Lemma \ref{lemma:sym} to extend the results for $\alpha > 1$. 
\bigskip

The rest of this section is organized as follows.  In the next subsection, i.e. in Section \ref{subsec:capob}, we obtain capacity outerbounds for the Gaussian $X$ channel. These bounds are analogous to those obtained for the deterministic $X$ channel in the previous section. In Section \ref{subsec:gdofob}, we translate the capacity outerbounds obtained in the next section to obtain a GDOF outerbound of Theorem \ref{thm:gdof}. In Section \ref{subsec:gdofach}, we use the insights obtained for the deterministic $X$ channel to show achievability of $d(\alpha)$ as described in Theorem \ref{thm:gdof}.

We remind the that while the achievable schemes we describe are valid for the symmetric case only, the capacity outerbounds shown in Section \ref{subsec:capob} are valid for the general setting.

\subsection{Outerbounds for the Gaussian $X$ channel}
\label{subsec:capob}
In this section, we study outerbounds for the $X$ channel. We first present known outer-bounds of the $X$ channel using previous works in the lemma below.
\begin{lemma}
\label{lem:knownob}
The rate tuple $(R_{11},R_{12},R_{21},R_{22})$ achieved by any reliable coding scheme over the $X$ channel satisfies the following bounds
\begin{eqnarray}
R_{ij} &\leq& \frac{1}{2} \log\big(1+{H_{ij}}^2 P_j\big) \label{bound1}\\
R_{1j}+R_{2j} &\leq&\frac{1}{2} \log\big(1+\max{(H_{1j}}^2,{H_{2j}}^2)P_j\big), j=1,2 \label{mac-bound1} \\
R_{i1}+R_{i2} &\leq&\frac{1}{2} \log\big(1+{H_{i1}}^2 P_1 +{H_{i2}}^2 P_2\big),i=1,2 \label{mac-bound2} \\
R_{11}+R_{22}+R_{12} &\leq&\frac{1}{2} \log\big(1+{H_{11}}^2 P_1+ {H_{12}}^2 P_2\big)+\frac{1}{2} \log\big(1+\frac{{H_{22}}^2 P_2}{1+{H_{12}}^2 P_2}\big) \label{z-bound1} \\
R_{22}+R_{11}+R_{21} &\leq&\frac{1}{2} \log\big(1+{H_{22}}^2 P_2+ {H_{21}}^2 P_1\big)+\frac{1}{2} \log\big(1+\frac{{H_{11}}^2 P_1}{1+{H_{21}}^2 P_1}\big) \label{z-bound2} \\
R_{11}+R_{12}+R_{21} & \leq &\frac{1}{2} \log\big(1+{H_{11}}^2 P_1+ {H_{12}}^2 P_2\big)+\frac{1}{2} \log\big( 1+\frac{{H_{21}}^2 P_1}{1+{H_{11}}^2 P_1}\big) \label{z-bound3} \\
R_{22}+R_{21}+R_{12} & \leq &\frac{1}{2} \log\big(1+{H_{22}}^2 P_2+ {H_{21}}^2 P_1\big)+\frac{1}{2} \log\big( 1+\frac{{H_{12}}^2 P_2}{1+{H_{22}}^2 P_2}\big)\label{z-bound4}
\end{eqnarray}
\end{lemma}
The bound in (\ref{bound1}) is trivial. (\ref{mac-bound1}) and (\ref{mac-bound2}) respectively follow from the bounds on the rates in the multiple access and broadcast channels contained in the $X$ channel. (\ref{z-bound1})-(\ref{z-bound4}) follow from the outerbound shown in \cite{Cadambe_Jafar_XFB} in the more general context of the $X$ channel with relays, feedback, noisy co-operation and full-duplex operation. For completeness we prove (\ref{z-bound1})-(\ref{z-bound4}) in Appendix \ref{append:zbound}.

In the following theorem we show Etkin-Tse-Wang bound for the Gaussian interference channel can be extended to the Gaussian $X$ channel. 
\begin{theorem}
\label{thm:etw_X}
The sum rate $R_{\Sigma}$ achieved by any reliable coding scheme over the $X$ channel satisfies the following bounds
\begin{eqnarray}
R_{\Sigma} &\leq&\frac{1}{2} \log\left(1+ {H_{12}}^2 P_2+ \frac{{H_{11}}^2 P_1}{1+{H_{21}}^2 P_1}\right)+ \frac{1}{2} \log\left(1+{H_{21}}^2 P_1+ \frac{{H_{22}}^2 P_2}{1+{H_{12}}^2 P_2} \right) \label{xoutbound:1} \\
R_{\Sigma} &\leq&\frac{1}{2} \log\left(1+ {H_{11}}^2 P_1+ \frac{{H_{12}}^2 P_2}{1+{H_{22}}^2 P_2}\right)+  \frac{1}{2}  \log\left(1+{H_{22}}^2 P_1+ \frac{{H_{21}}^2 P_2}{1+{H_{11}}^2 P_1} \label{xoutbound:2} \right)
\end{eqnarray}
\end{theorem}
\begin{proof}
\begin{figure}
\begin{center}\setlength{\unitlength}{0.00054167in}
\begingroup\makeatletter\ifx\SetFigFont\undefined%
\gdef\SetFigFont#1#2#3#4#5{%
  \reset@font\fontsize{#1}{#2pt}%
  \fontfamily{#3}\fontseries{#4}\fontshape{#5}%
  \selectfont}%
\fi\endgroup%
{\renewcommand{\dashlinestretch}{30}
\begin{picture}(7596,5029)(0,-10)
\path(5175,3958)(5175,3808)
\path(5100,3883)(5250,3883)
\path(5175,1183)(5175,1033)
\path(5100,1108)(5250,1108)
\put(2550,3883){\ellipse{336}{336}}
\put(2475,1108){\ellipse{336}{336}}
\put(5175,1108){\ellipse{336}{336}}
\put(5175,3883){\ellipse{336}{336}}
\path(1800,1108)(2325,1108)
\path(2205.000,1078.000)(2325.000,1108.000)(2205.000,1138.000)
\path(1875,3883)(2400,3883)
\path(2280.000,3853.000)(2400.000,3883.000)(2280.000,3913.000)
\path(6300,3883)(6825,3883)
\path(6705.000,3853.000)(6825.000,3883.000)(6705.000,3913.000)
\path(2700,3883)(5025,3883)
\blacken\path(4905.000,3838.000)(5025.000,3883.000)(4905.000,3928.000)(4941.000,3883.000)(4905.000,3838.000)
\path(2700,3883)(5025,1183)
\blacken\path(4912.598,1244.569)(5025.000,1183.000)(4980.797,1303.296)(4970.188,1246.653)(4912.598,1244.569)
\path(2625,1108)(5025,1108)
\blacken\path(4905.000,1063.000)(5025.000,1108.000)(4905.000,1153.000)(4941.000,1108.000)(4905.000,1063.000)
\path(5325,3883)(5850,3883)
\path(5730.000,3853.000)(5850.000,3883.000)(5730.000,3913.000)
\path(6225,1108)(6750,1108)
\path(6630.000,1078.000)(6750.000,1108.000)(6630.000,1138.000)
\path(900,1108)(1425,1108)
\path(1305.000,1078.000)(1425.000,1108.000)(1305.000,1138.000)
\path(975,3883)(1500,3883)
\path(1380.000,3853.000)(1500.000,3883.000)(1380.000,3913.000)
\path(2625,1108)(5025,3808)
\blacken\path(4978.910,3688.415)(5025.000,3808.000)(4911.643,3748.207)(4969.193,3745.218)(4978.910,3688.415)
\path(5175,3133)(5175,3733)
\blacken\path(5220.000,3613.000)(5175.000,3733.000)(5130.000,3613.000)(5175.000,3649.000)(5220.000,3613.000)
\blacken\path(5130.000,1378.000)(5175.000,1258.000)(5220.000,1378.000)(5175.000,1342.000)(5130.000,1378.000)
\path(5175,1258)(5175,1858)
\dashline{60.000}(6525,4633)(6525,3883)
\path(6495.000,4003.000)(6525.000,3883.000)(6555.000,4003.000)
\path(6555.000,988.000)(6525.000,1108.000)(6495.000,988.000)
\dashline{60.000}(6525,1108)(6525,358)
\path(5325,1108)(5850,1108)
\path(5730.000,1078.000)(5850.000,1108.000)(5730.000,1138.000)
\put(5925,1033){\makebox(0,0)[lb]{{\SetFigFont{9}{10.8}{\rmdefault}{\mddefault}{\updefault}$Y_2$}}}
\put(6900,3808){\makebox(0,0)[lb]{{\SetFigFont{9}{10.8}{\rmdefault}{\mddefault}{\updefault}$\hat{W}_{11},\hat{W}_{12}$}}}
\put(5925,3808){\makebox(0,0)[lb]{{\SetFigFont{9}{10.8}{\rmdefault}{\mddefault}{\updefault}$Y_1$}}}
\put(6825,1033){\makebox(0,0)[lb]{{\SetFigFont{9}{10.8}{\rmdefault}{\mddefault}{\updefault}$\hat{W}_{21},\hat{W}_{22}$}}}
\put(1575,3808){\makebox(0,0)[lb]{{\SetFigFont{9}{10.8}{\rmdefault}{\mddefault}{\updefault}$X_1$}}}
\put(1500,1033){\makebox(0,0)[lb]{{\SetFigFont{9}{10.8}{\rmdefault}{\mddefault}{\updefault}$X_2$}}}
\put(5025,2008){\makebox(0,0)[lb]{{\SetFigFont{8}{9.6}{\rmdefault}{\mddefault}{\updefault}$Z_2$}}}
\put(5025,2833){\makebox(0,0)[lb]{{\SetFigFont{8}{9.6}{\rmdefault}{\mddefault}{\updefault}$Z_1$}}}
\put(6750,583){\makebox(0,0)[lb]{{\SetFigFont{8}{9.6}{\rmdefault}{\mddefault}{\updefault}Genie}}}
\put(6600,4333){\makebox(0,0)[lb]{{\SetFigFont{8}{9.6}{\rmdefault}{\mddefault}{\updefault}Genie}}}
\put(3450,4033){\makebox(0,0)[lb]{{\SetFigFont{8}{9.6}{\rmdefault}{\mddefault}{\updefault}$H_{11}$}}}
\put(3225,3283){\makebox(0,0)[lb]{{\SetFigFont{8}{9.6}{\rmdefault}{\mddefault}{\updefault}$H_{21}$}}}
\put(3375,733){\makebox(0,0)[lb]{{\SetFigFont{8}{9.6}{\rmdefault}{\mddefault}{\updefault}$H_{22}$}}}
\put(2925,2008){\makebox(0,0)[lb]{{\SetFigFont{8}{9.6}{\rmdefault}{\mddefault}{\updefault}$H_{12}$}}}
\put(6150,58){\makebox(0,0)[lb]{{\SetFigFont{8}{9.6}{\rmdefault}{\mddefault}{\updefault}$W_{12},S_{12}$}}}
\put(6150,4858){\makebox(0,0)[lb]{{\SetFigFont{8}{9.6}{\rmdefault}{\mddefault}{\updefault}$W_{21},S_{21}$}}}
\put(150,3808){\makebox(0,0)[lb]{{\SetFigFont{9}{10.8}{\rmdefault}{\mddefault}{\updefault}$W_{11},W_{21}$}}}
\put(0,1033){\makebox(0,0)[lb]{{\SetFigFont{9}{10.8}{\rmdefault}{\mddefault}{\updefault}$W_{12},W_{22}$}}}
\end{picture}
}
\caption{Genie aided $X$ channel used in proof of Theorem \ref{thm:etw_X}}
\end{center}
\end{figure}
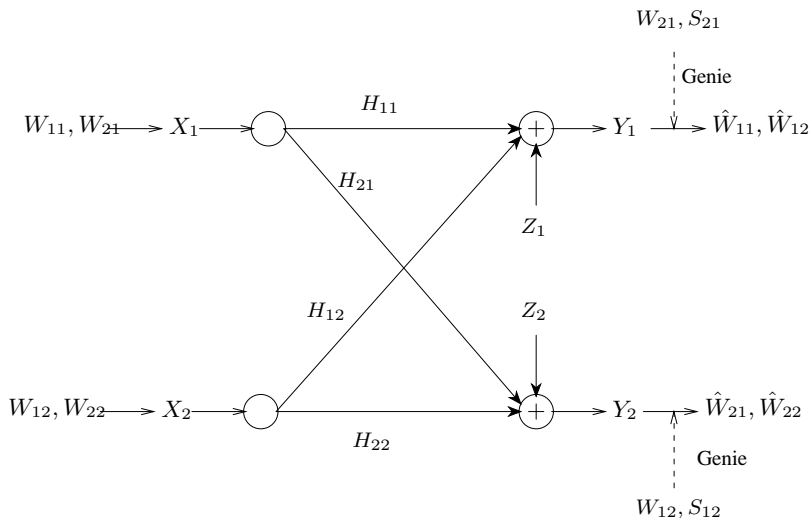

Let
\begin{eqnarray*}
S_{ij}(t) &=& H_{ij} X_{j}(t) + Z_{i}(t), i,j \in \{1,2\}
\end{eqnarray*}
Note that $S_{ij}$ are auxiliary variables similar to the ETW outerbound of the interference channel.
Consider any reliable coding scheme. Now, let a genie provide ${S_{12}}^{(T)}= {H_{12}} {X_{2}}^{(T)} + {Z_{1}}^{(T)}$ and $W_{12}$ to receiver $2$.
From Fano's inequality, for any codeword of length $T$, we can write
\begin{eqnarray}
T(R_{22}+R_{21}  - \epsilon) &\leq& I\left(W_{22}, W_{21}; {Y_{2}}^{(T)}, {S_{12}}^{(T)},W_{12}\right)\\
&\leq& I\left( W_{22}, W_{21}; W_{12}\right) + I\left(W_{22}, W_{21}; {Y_{2}}^{(T)}, {S_{12}}^{(T)}~|~ W_{12}\right) \label{Rx2:ineq} \\
&\leq& h\left({Y_{2}}^{(T)}, {S_{12}}^{(T)}~|~ W_{12}\right)  - h\left({Y_{2}}^{(T)}, {S_{12}}^{(T)}|~W_{21}, W_{12}, W_{22}\right) \label{Rx2:ineq1}\\
&\leq& h\left({S_{12}}^{(T)}~|~ W_{12}\right) + h\left({Y_{2}}^{(T)}~|~ {S_{12}}^{(T)}, W_{12}\right)  - h\left({Y_{2}}^{(T)}~|~ W_{12}, W_{22}, W_{21}\right) \nonumber \\&&-h\left({S_{12}}^{(T)}~|~ {Y_{2}}^{(T)}, W_{12}, W_{22}, W_{21}\right)\label{Rx2:ineq2}\\
&\leq& h\left({S_{12}}^{(T)}~|~ W_{12}\right) + h\left({Y_{2}}^{(T)}~|~ {S_{12}}^{(T)}\right)  - h\left({Y_{2}}^{(T)}~|~ {X_{2}}^{(T)},W_{12}, W_{22},  W_{21}, \right) \nonumber \\&&-h\left({S_{12}}^{(T)}~|~{Y_{2}}^{(T)}, {X_{1}}^{(T)}, {X_{2}}^{(T)}, W_{12}, W_{22}, W_{21}\label{Rx2:ineq3}\right)\\
&\leq& h\left({S_{12}}^{(T)}~|~ W_{12}\right) + h\left({Y_{2}}^{(T)}~|~ {S_{12}}^{(T)}\right) - h\left({H_{21}} {X_{1}}^{(T)}+{Z_{2}}^{(T)}~|~ W_{12}, W_{22}, W_{21}, {X_{2}}^{(T)}\right) \\&&- h\left({Z_{1}}^{(T)}~|~ {Y_{2}}^{(T)}, {X_{1}}^{(T)},X_2^{(T)},W_{12}, W_{22}, W_{21}, W_{21}, \right)\label{Rx2:ineq4}\\
&\leq& h\left({S_{12}}^{(T)}~|~ W_{12}\right) + h\left({Y_{2}}^{(T)}~|~ {S_{12}}^{(T)}\right)  - h\left({S_{21}}^{(T)}~|~W_{21}\right) - h\left({Z_{1}}^{(T)}\right) \label{Rx2:ineq5}
\end{eqnarray}
In (\ref{Rx2:ineq}), the first summand is zero because of $W_{12}$ is independent of $W_{21},W_{22}$. We have used the chain rule in (\ref{Rx2:ineq1}). In (\ref{Rx2:ineq2}), we have used the fact that conditioning does not reduce the entropy on the second, third and fourth summands on the right hand side.  In (\ref{Rx2:ineq5}), we have used the fact that $S_{21}^{(T)}={H_{21}}{X_{1}}^{(T)}+{Z_{2}}^{(T)}$ is independent of messages $W_{12},W_{22}$ and the codeword  ${X_{2}}^{(T)}$.

Similarly, if a genie provides receiver $1$ with $S_{21}^{(T)}$ and $W_{21}$, we can bound rates at receiver $1$ as
\begin{equation} T R_{12}+T R_{11} -T \epsilon \leq h\left({S_{21}}^{(T)}| W_{21}\right) + h\left({Y_{1}}^{(T)}| {S_{21}}^{(T)}\right)  - h\left({S_{12}}^{(T)}|W_{12}\right) - h\left({Z_{1}}^{(T)}\right) \label{Rx1:outerbound}\end{equation}
Adding  (\ref{Rx1:outerbound}) and (\ref{Rx2:ineq5}), we get
\begin{eqnarray*} T (R_{\Sigma}-\epsilon) &\leq& h\left({Y_{1}}^{(T)}~| {S_{21}}^{(T)}\right) + h\left({Y_{2}}^{(T)}|{S_{12}}^{(T)}\right) - h\left({Z_{1}}^{(T)}\right) - h\left({Z_{2}}^{(T)}\right) \\
T (R_{\Sigma}-\epsilon) &\leq& \sum_{t=1}^{T} \left[h\left(Y_{1}(t)| S_{21}(t)\right) + h\left(Y_{2}(t)|S_{12}(t)\right)\right] - T h\left({Z_{1}}\right) - T h\left({Z_{2}}\right)
	\end{eqnarray*}
The second inequality above uses the chain rule combined with fact that conditioning does not increase entropy.  Therefore, dividing by $T$, taking $T \to \infty$, and using the fact that Gaussian variables maximize conditional entropy, we get
$$ R_{\Sigma} \leq \frac{1}{2} \log\left(1+ {H_{12}}^2 P_2+ \frac{{H_{11}}^2 P_1}{1+{H_{21}}^2 P_1}\right)+ \frac{1}{2} \log\left(1+{H_{21}}^2 P_1+ \frac{{H_{22}}^2 P_2}{1+{H_{12}}^2 P_2} \right)$$
Note that the above bound on the sum capacity is identical to the ETW bound for the sum capacity of the weak interference channel \cite{Etkin_Tse_Wang}.
Furthermore, we can get another bound on the sum capacity of the $X$ channel, symmetric to the above bound by allowing a genie to provide $S_{22}^{(T)}, W_{11}$ to receiver $1$ and  $S_{11}^{(T)}, W_{22}$ to receiver $2$. In this case, we get
$$ R_{\Sigma} \leq \frac{1}{2} \log\left(1+ {H_{11}}^2 P_1+ \frac{{H_{12}}^2 P_2}{1+{H_{22}}^2 P_2}\right)+ \frac{1}{2} \log\left(1+{H_{22}}^2 P_2+ \frac{{H_{21}}^2 P_1}{1+{H_{11}}^2 P_1} \right)$$
\end{proof}

\subsection{Generalized Degrees of Freedom Outerbound}
\label{subsec:gdofob}
We now translate the capacity outerbounds stated above to a generalized degrees of freedom outerbound. We only find an GDOF outerbound for $\alpha \leq 1$ below, since for $\alpha \geq 1$, we can use lemma \ref{lemma:sym} along with the following theorem to bound the GDOF.
\begin{theorem}
The GDOF of the $X$ channel for $\alpha \leq 1 $ can be bounded as
$$ d(\alpha) \leq 2 \min\left(\max(\alpha,1-\alpha), 1-\alpha/3\right)$$
\end{theorem}
\bigskip
\begin{proof}
Now, for the $X$ channel where
${H_{11}}={H_{22}}=\sqrt{\rho}$, ${H_{21}}={H_{12}}=\sqrt{\rho^{\alpha}}$, and $P_1 = P_2 = 1$ the outerbound in (\ref{xoutbound:1}) leads to

$$R_{\Sigma}(\rho, \alpha) \leq \log(1+\rho^{\alpha}+\frac{\rho}{1+\rho^\alpha})$$
Dividing the above inequality by $\frac{1}{2} \log \rho$ and then taking limits as $\rho \to \infty$, we get
\begin{equation} d(\alpha) \leq 2 \max(\alpha,1-\alpha) \label{gdofbound:1} \end{equation}

Similarly, the bounds in (\ref{z-bound1})-(\ref{z-bound4}) lead to
\begin{eqnarray*}
R_{11}+R_{12}+R_{22} &\leq& \frac{1}{2} \log\big(1+ \rho + \rho^{\alpha}\big) + \frac{1}{2} \log\big(1+\frac{\rho}{1+\rho^\alpha}\big) \\
R_{22}+R_{11}+R_{21} &\leq& \frac{1}{2} \log\big(1+\rho+\rho^{\alpha}\big)+\frac{1}{2} \log\big(1+\frac{\rho}{1+\rho^{\alpha}}\big)  \\
R_{11}+R_{12}+R_{21} & \leq &\frac{1}{2} \log\big(1+\rho + \rho^\alpha \big)+\frac{1}{2} \log\big( 1+\frac{\rho^{\alpha}}{1+\rho}\big)  \\
R_{22}+R_{21}+R_{12} & \leq &\frac{1}{2} \log\big(1+\rho+\rho^{\alpha} \big)+\frac{1}{2} \log\big( 1+\frac{\rho^{\alpha}}{1+\rho}\big)
\end{eqnarray*}

Adding, we get
\begin{eqnarray*} 3(R_{11}+R_{12}+R_{21}+R_{22}) \leq 2 \log(1+\rho+\rho^{\alpha})+\log(1+\frac{\rho}{1+\rho^\alpha})+ \log(1+\frac{\rho^\alpha}{1+\rho})\\
		\Rightarrow 3 C_{\Sigma}(\rho, \alpha) \leq 4 \log(1+\rho+\rho^{\alpha})-\log(1+\rho)- \log(1+\rho^\alpha)\\
	\end{eqnarray*}
Dividing the above inequality by $\frac{1}{2} \log \rho$ and then taking limits as $\rho \to \infty$, we get
\begin{equation}
d(\alpha) \leq  2 - \frac{2 \alpha}{3}\label{gdofbound:2}
\end{equation}

Therefore, from (\ref{gdofbound:1}), (\ref{gdofbound:2}), we get
$$ d(\alpha) \leq 2\min\left(\max(\alpha,1-\alpha),1 - \frac{ \alpha}{3}\right)$$
\end{proof}

\subsection{Achievability of Generalized Degrees of Freedom}
\label{subsec:gdofach}
In this section, we provide an outline of the proof for the achievability of Theorem \ref{thm:gdof}. The main idea of the proof is to transform the symmetric Gaussian $X$ channel to the deterministic $X$ channel by imposing some structure on the transmit signal. Then, we apply the achievable scheme derived in the previous subsection to obtain the GDOF result of the Gaussian case. The proof follows the similar arguments used in \cite{Jafar_Vishwanath_GDOF,Cadambe_Jafar_Shamai}, and we include an outline of the proof for the sake of the completeness. We only consider the case that $0 \leq \alpha \leq 1$ here. The GDOF characteriation for $ \alpha >1 $ follows from Lemma \ref{lemma:sym}.
The symmetric Gaussian $X$ channel is defined by (\ref{eqn:syminout1}),(\ref{eqn:syminout2}). 

For a given $\alpha \in [0,1]$, we can find a pair of non-negative integers $(n_d, n_c)$ and a very small nonnegative value $\epsilon$ such that
\begin{eqnarray}
\label{eqn:alpah}
\alpha = \frac{1}{n_d} \left( n_c + \epsilon (n_d - n_c) \right).
\end{eqnarray}
Note that when $\alpha$ is a rational number, $\epsilon$ is chosen to be zero. But when $\alpha$ is not rational, $\epsilon (n_d - n_c)/n_d$ is used to compensate the difference between $\alpha$ and a rational number $\frac{n_c}{n_d}$ that is very close to $\alpha$. Also note that $(n_c, n_d)$ is chosen such that (\ref{eqn:gdof}) can be achieved without symbol extension for the symmetric deterministic channel with parameter $(n_c, n_d)$.

Consider the sequence of channels, i.e. $\rho$ indexed by $N$ such that
\begin{eqnarray}
\label{eqn:rho}
\rho = Q^{\frac{2Nn_d}{1-\epsilon}}
\end{eqnarray}
where $Q$ is a very large but fixed positive integer and $N$ is a positive integer whose value grows to infinity. Note that $\rho$ grows to infinity as $N$ grows to infinity.

For this channel, we describe the GDOF optimal achievable scheme, for a given $\alpha$ below.

\textit{1) Transmit Scheme:} We impose the following structure on the Q-ary representation of the transmit signal $X_{i}$ at transmitter $i$ for $i \in \{1,2\}$. 
\begin{eqnarray}
\label{eqn:XStruct}
X_{i}
&=& \frac{1}{\sqrt{\rho}}\sum_{k=0}^{Nn_d-1}x_{i,k}Q^k
\end{eqnarray}
In other words, the Q-ary representation of the transmit signal $X_{i}$ looks like $(x_{i,Nn_d - 1}  x_{i,Nn_d - 2}  \ldots  x_{i,2}  x_{i,1}  x_{i,0} . 0 0 \ldots)_{Q}$
 The values of $x_{i,k}$ are restricted to the set $\{ 1,\ldots,\lfloor \frac{Q-1}{4} - 1 \rfloor \}$ to ensure that addition of interference does not produce carry over.
 Since $\epsilon$ is a small non-negative value, we have
\begin{eqnarray}
\mathbf{E} \left[ X_{i}^2\right]&=& \mathbf{E} \left[ \frac{1}{\rho} \left(\sum_{k=0}^{Nn_d-1} x_{i,k} Q^k\right)^2 \right]\\
&\leq &  \frac{1}{\rho} \mathbf{E} \left[ \sum_{k=0}^{Nn_d-1}  (Q-1) Q^{k} \right]\\
&\leq &  \frac{Q^{2Nn_d}}{\rho}\\
& \leq & 1
\end{eqnarray}
Thus, the encoding scheme satisfies the power constraint.

Since the achievable scheme developed in the previous subsection also works in $\mathcal{F}_{\lfloor \frac{Q-1}{4}\rfloor -2}^{Nn_d}$, we can use it to find the transmit signals $\mathbf{X}_{1}, \mathbf{X}_{2} \in \mathcal{F}_{\lfloor \frac{Q-1}{4}\rfloor -2}^{Nn_d}$ for the symmetric deterministic channel with parameter $(n_c,n_d)$ and then obtain the corresponding $X_{1},X_{2} \in \mathbb{R}^+$ by
\begin{eqnarray}
X_{i}
= \frac{1}{\sqrt{\rho}}
\left[ \begin{array}{cccccc}
Q^{Nn_d - 1} & Q^{Nn_d - 2} & \cdots & Q^2 &
Q^1 & 1
\end{array} \right]
\mathbf{X}_{i}
\end{eqnarray}
where the last term $\mathbf{X}_{i}$ is the $Nn_d \times 1$ transmit vector for the deterministic channel.

\textit{2) Receive Scheme :} Each receiver takes the magnitude of the received signal, reduces to modulo $Q^{Nn_d}$, discards the value below the decimal point, and expresses the result in Q-ary representation as
\begin{eqnarray}
\overline{Y}_{i}
&=&
\left\lfloor
\begin{array}{ccc}
\left| Y_{i} \right| & \textrm{mod} & Q^{Nn_d}
\end{array}
\right\rfloor \\
&=&
\begin{array}{cc}
\displaystyle{
\sum_{k=0}^{Nn_d-1}} y_{i,k}Q^k,
&
y_{k,i} \in \{ 0,1,\ldots,Q-1\}
\end{array}
\end{eqnarray}
Substituting (\ref{eqn:alpah}) and (\ref{eqn:rho}) into (\ref{eqn:syminout1}) and (\ref{eqn:syminout2}), we can rewrite the input output equation as
\begin{eqnarray}
\begin{array}{cc}
\label{eqn:NewSymInout}
Y_{i} = \overline{X}_{i} +  Q^{N(n_c - n_d)} \overline{X}_{j} + Z_{i},
&
(i,j) \in \left\{(1,2),(2,1) \right\}
\end{array}
\end{eqnarray}
where 
$$ \overline{X}_i \define \sqrt{\rho} X_i$$
Note that multiplication by $Q^{N(n_c - n_d)}$ shifts the decimal point in the Q-ary representation of $\overline{X}_{j}$ by $N(n_d - n_c)$ places to the left. Therefore, in the absence of noise, the $Nn_d$ digits of $\overline{X}_{1}, \overline{X}_{2}, \overline{Y}_{1}$, and $\overline{Y}_{2}$ behave exactly like the symmetric deterministic channel with parameter $(Nn_c, Nn_d)$. Next, we will consider the effect of AWGN. Let $P_k^e$ be the probability that
\begin{eqnarray}
\overline{Y}_k^{i} \neq \overline{X}_k^{i} + \overline{X}_{k+N(n_d-n_c)}^{j}
\end{eqnarray}
happens for any $(i,j)\in \{(1,2),(2,1) \}$. Due to fact that any additive noise with magnitude no greater than $Q^{k-1}$ does not affect the coefficient of $Q^{k}$, we have
\begin{eqnarray}
1 - P_k^e \geq \textrm{Prob}
\left(\left| Z_{1} \right| \leq Q^{k-1}, \left| Z_{2} \right| \leq Q^{k-1} \right).
\end{eqnarray}
Thus, $P_k^e$ monotonically decreases to $0$ as $k$ grows to infinity. A key result of this observation is that the multi-level coding approach \cite{Cadambe_Jafar_Shamai} approximates the deterministic channel within $o(N)$. Thus, we have
\begin{eqnarray}
R_\Sigma
&=&
R_{\Sigma,det}(Nn_c, Nn_d)
\log_Q \left( \left\lfloor \frac{Q-1}{4} - 2 \right\rfloor \right) + o(N) \\
&=&
\label{eqn:Rsum}
NR_{\Sigma,det}(n_c, n_d)
\log_Q \left( \left\lfloor \frac{Q-1}{4} - 2 \right\rfloor \right) + o(N)
\end{eqnarray}
Combining (\ref{eqn:rho}), (\ref{eqn:Rsum}), and (\ref{eqn:gdofDef}), we have
\begin{eqnarray}
d(\alpha)
&\geq&
\limsup_{N \rightarrow \infty}
\frac{NR_{\Sigma,det} \left( n_d \left( \frac{\alpha - \epsilon}{1 - \epsilon} \right), n_d \right)
\log_Q \left( \left\lfloor \frac{Q-1}{4} - 2 \right\rfloor \right) + o(N)}{\frac{Nn_d}{1-\epsilon}} \\
&=&
\frac{1-\epsilon}{n_d} R_{\Sigma,det}\left( n_d \left( \frac{\alpha - \epsilon}{1 - \epsilon} \right), n_d \right)
\log_Q \left( \left\lfloor \frac{Q-1}{4} - 2 \right\rfloor \right)
\end{eqnarray}
Carrying out the substitution of $R_{\Sigma,det}(\cdot,\cdot)$, choosing $Q$ and $\epsilon$ to be arbitrarily large and small respectively, and comparing with the outerbound, we finish the proof of Theorem \ref{thm:gdof}.

\section{Capacity of the Noisy $X$ Channel}\label{sec:noisyX}
We state the result for the general (asymmetric) case as follows.
\begin{theorem}
\label{thm:AsymNoisyX}
If
\begin{eqnarray}
\left| \frac{H_{12}}{H_{22}} \left( 1 + H_{21}^2 P_1 \right) \right| +
\left| \frac{H_{21}}{H_{11}} \left( 1 + H_{12}^2 P_2 \right) \right| \leq 1, \label{eq:noisyint1}
\end{eqnarray}
then the sum capacity of the Gaussian $X$ channel is given by
\begin{eqnarray}
C_{\Sigma} =
\frac{1}{2} \log \left( 1 + \frac{H_{11}^2 P_1}{1 + H_{12}^2 P_2} \right) +
\frac{1}{2} \log \left( 1 + \frac{H_{22}^2 P_2}{1 + H_{21}^2 P_1} \right).
\end{eqnarray}
Similarly, if
\begin{eqnarray}
\label{eqn:prt2cndt}
\left| \frac{H_{22}}{H_{12}} \left( 1 + H_{11}^2 P_1 \right) \right| +
\left| \frac{H_{11}}{H_{21}} \left( 1 + H_{22}^2 P_2 \right) \right| \leq 1, \label{eq:noisyint2}
\end{eqnarray}
then the sum capacity of the Gaussian $X$ channel is given by
\begin{eqnarray}
\label{eqn:prt2cpcty}
C_{\Sigma} =
\frac{1}{2} \log \left( 1 + \frac{H_{21}^2 P_1}{1 + H_{22}^2 P_2} \right) +
\frac{1}{2} \log \left( 1 + \frac{H_{12}^2 P_2}{1 + H_{11}^2 P_1} \right).
\end{eqnarray}
\end{theorem}


\begin{proof}
Let 
$$ \tilde{S}_i(t) = X_i(t) + \tilde{Z}_i(t),~i=1,2$$
$\tilde{Z}_i$ is white Gaussian with zero mean and variance $\sigma_i^2$. Also, let $\tilde{Z}_i(t)$ be correlated with $Z_i(t)$ as $$E\left[Z_i(t) \tilde{Z}_i (t) \right] = \sigma_i \eta_i.$$ Let a genie provide $\tilde{S}_1$ to receiver $1$ and $\tilde{S}_2$ to receiver $2$. Now, we can write using Fano's inequality for a codeword spanning $T$ symbols,
\begin{eqnarray}
T(R_{22}+ R_{21}-\epsilon)
&\leq& I\left(W_{22}, W_{21}; Y_2^{(T)}, \tilde{S}_{2}^{(T)}~|~ W_{12}\right)\\
&\leq& h\left(Y_2^{(T)}, \tilde{S}_{2}^{(T)}~|~ W_{12}\right) - h\left(Y_2^{(T)},\tilde{S}_{2}^{(T)}~|~ W_{12}, W_{22}, W_{21} \right) \\
&\leq& h\left(\tilde{S}_{2}^{(T)}~|~ W_{12}\right) + h\left(Y_2^{(T)}~|~ \tilde{S}_{2}^{(T)} \right)  - h\left(\tilde{S}_2^{(T)}~|~ W_{12}, W_{22}, W_{21}, X_2^{(T)}\right) \nonumber \\&&- h\left(Y_{2}^{(T)}~|~ \tilde{S}_2^{(T)}, W_{12}, W_{22}, W_{21},X_2^{(T)}\right)\label{eq:conditioning}\\
&\leq& h\left(\tilde{S}_{2}^{(T)}~|~ W_{12}\right) + h\left(Y_2^{(T)}~|~ \tilde{S}_{2}^{(T)} \right)  - h\left(X_2^{(T)}+\tilde{Z}_2^{(T)}~|~ W_{12}, W_{22}, W_{21},X_2^{(T)}\right) \nonumber \\&&- h\left(H_{21} X_1^{(T)} + H_{22} X_2^{(T)} + Z_2^{(T)}~|~ Z_2^{(T)}, W_{12}, W_{22}, W_{21}, \tilde{S}_2^{(T)},X_2^{(T)}\right)\\
&\leq& h\left(\tilde{S}_{2}^{(T)}~|~ W_{12}\right) + h\left(Y_2^{(T)}~|~ \tilde{S}_{2}^{(T)} \right)  - h\left(H_{21}X_1^{(T)}+Z_2^{(T)}~|~ W_{21}, \tilde{Z}_2{(T)}\right) - h\left(\tilde{Z}_2^{(T)} \label{eq:R22R21bound} \right)
\end{eqnarray}
where (\ref{eq:conditioning}) holds because we have applied the fact that conditioning reduces entropy in the second,third and fourth terms. In (\ref{eq:R22R21bound}) we have used the fact that $X_{2}^{(T)}, W_{22}, W_{12}$ are independent of $X_{1}^{(T)}$, $Z_1^{(T)}$ and $\tilde{Z}_1^{(T)}$.

Similarly, we can bound rates $R_{12}$ and $R_{11}$ as 
\begin{eqnarray}
T(R_{12}+ R_{11}-\epsilon) & \leq & h\left(\tilde{S}_{1}^{(T)}~|~ W_{21}\right) + h\left(Y_1^{(T)}~|~ \tilde{S}_{1}^{(T)} \right)  - h\left(H_{12}X_2^{(T)}+Z_1^{(T)}~|~ W_{12}, \tilde{Z}_1{(T)}\right) - h\left(\tilde{Z}_1^{(T)}\right) \label{eq:R12R11bound}
\end{eqnarray}

Adding (\ref{eq:R22R21bound}) and (\ref{eq:R12R11bound})  we get
\begin{eqnarray*} T(R_{11}+R_{12}+R_{21}+R_{22} -\epsilon) &\leq&  \underbrace{h\left(X_{2}^{(T)}+\tilde{Z}_2^{(T)}~|~ W_{12}\right) - h(H_{12} X_{2}^{(T)} + Z_1^{(T)}~|~W_{12}, \tilde{Z}_1^{(T)})}_{U_1}  \\ && + 
\underbrace{h\left(X_{1}^{(T)}+\tilde{Z}_1^{(T)}~|~ W_{21}\right) - h(H_{21} X_{1}^{(T)} + Z_2^{(T)}~|~W_{21}, \tilde{Z}_2^{(T)}) }_{U_2} \\ && + 
\underbrace{h(Y_1^{(T)}| S_1^{(T)}) + h(Y_2^{(T)}| S_2^{(T)}) - h(\tilde{Z}_1^{(T)}) - h(\tilde{Z}_2^{(T)})}_{U_{3}}
	\end{eqnarray*}
The rest of the proof goes along the same lines as described in \cite{Sreekanth_Veeravalli}. We only highlight the differences here. 

We first notice that $U_3$ is maximized if we choose $X_{1}$ to have a Gaussian distribution, since conditional entropy $h(Y_i^{(T)}|S_i^{(T)})$ is maximized by the Gaussian distribution.  Therefore, we can write
$$ U_{3} \leq {h(Y_{1G}^{(T)}| \tilde{S}_{1G}^{(T)}) + h(Y_{2G}^{(T)}| \tilde{S}_{2G}^{(T)}) - h(\tilde{Z}_1^{(T)}) - h(\tilde{Z}_2^{(T)})}$$
where for $i=1,2$, $X_{iG}^{(T)},Y_{iG}^{(T)}, \tilde{S}_{iG}^{(T)}$ are variables obtained by using a Gaussian i.i.d sequence of power $P_i$ for $X_i$. 
Now, following the proof of \cite{Sreekanth_Veeravalli}, we derive conditions on $\eta_i, \sigma_i$ for $i=1,2$ so that circularly symmetric Gaussian distribution on $X_i^{(T)}$ maximizes $U_1$ and $U_2$ as well. Specifically, we choose 
\begin{eqnarray} 
\sigma_1^2 &\leq& \frac{1-\eta_2^2}{H_{21}^2} \label{eq:worstnoisecond}
\end{eqnarray}
and circularly symmetric independent Gaussian variables $V^{(T)},V_1^{(T)}$ such that $V \sim \mathcal{N}(0, 1-\eta_2^2 - \sigma_1^2)$ and $V_1 \sim \mathcal{N}(0,\sigma_1^2)$.
Now, we observe that 
\begin{eqnarray*}
U_{2} &=& h\left(X_{1}^{(T)}+\tilde{Z}_1^{(T)}~|~ W_{21}\right) - h\left(H_{21} X_{1}^{(T)} + Z_2^{(T)}~|~W_{21}, \tilde{Z}_2^{(T)}\right) \\
&=&  -I\left(V^{(T)}; X_{1}^{(T)}+V^{(T)}+V_1^{(T)}|W_{21}\right) \\
&=& -h(V^{(T)}|W_{12}) + h(V^{(T)}|X_{1}^{(T)}+V^{(T)}+V_1^{(T)},W_{12}) \\
&\stackrel{(a)}{\leq}& -h(V^{(T)}) + h(V^{(T)}|X_{1}^{(T)}+V^{(T)}+V_1^{(T)}) \\
&\leq& - I\left(V^{(T)}:X_{1}^{(T)}+V^{(T)}+V_1^{(T)} \right)  \label{eq:worstnoiseapp}\\
&\stackrel{(b)}{\leq}& - I\left(V^{(T)}:X_{1G}^{(T)}+V^{(T)}+V_1^{(T)} \right)  \\
&\leq& h\left(X_{1G}^{(T)}+\tilde{Z}_1^{(T)}~\right) - h\left(H_{21} X_{1G}^{(T)} + Z_2^{(T)}~|, \tilde{Z}_2^{(T)}\right) 
\end{eqnarray*}
 In the first term of the summand in $(a)$, we have used the fact that $V^{(T)}$ is independent of $W_{12}$ and in the second summand of $(a)$, we have used the fact that conditioning reduces entropy. Inequality $(b)$ holds because of the worst case noise lemma \cite{Diggavi} as long as  (\ref{eq:worstnoisecond}) is satisfied.
Along the same lines, by choosing 
\begin{eqnarray} 
\sigma_2^2 &\leq& \frac{1-\eta_1^2}{H_{12}^2} \label{eq:worstnoisecond2}
	\end{eqnarray}
we can bound $U_1$ in a similar manner.
Therefore, we can write 
\begin{eqnarray*} T(R_{11}+R_{12}+R_{21}+R_{22}-\epsilon) & \leq & h\left(X_{1G}^{(T)}+\tilde{Z}_1^{(T)}\right) - h\left(H_{21} X_{1G}^{(T)} + Z_2^{(T)}~|~ \tilde{Z}_2^{(T)}\right) \\
&&+h\left(X_{2G}^{(T)}+\tilde{Z}_2^{(T)}\right) - h\left(H_{12} X_{2G}^{(T)} + Z_1^{(T)}~| \tilde{Z}_1^{(T)}\right) \\
&&+h\left(Y_{1G}^{(T)}| \tilde{S}_{1G}^{(T)}\right) + h\left(Y_{2G}^{(T)}| \tilde{S}_{2G}^{(T)}\right) - h\left(\tilde{Z}_{1}^{(T)}\right) - h\left(\tilde{Z}_2^{(T)}\right)\\
&\leq& I\left(X_{1G}^{(T)};Y_{1G}^{(T)},\tilde{S}_{1G}^{(T)}\right)+ I\left(X_{2G}^{(T)};Y_{2G}^{(T)},\tilde{S}_{2G}^{(T)}\right)
\end{eqnarray*}
The rest of the proof follows Lemma 10 in \cite{Sreekanth_Veeravalli}. Specifically, it can be shown that if 
\begin{eqnarray}
H_{11}\sigma_1\eta_1 &=& H_{12}^2 P_2 + 1 \label{eq:noisyintcond1}\\
H_{22}\sigma_2\eta_2 &=& H_{21}^2 P_1 + 1 \label{eq:noisyintcond2}
\end{eqnarray}
then, 
\begin{eqnarray*}
I\left(X_{1G}^{(T)};Y_{1G}^{(T)},\tilde{S}_{1G}^{(T)}\right)  &=& I\left(X_{1G}^{(T)};Y_{1G}^{(T)}\right)\\
I\left(X_{2G}^{(T)};Y_{2G}^{(T)},\tilde{S}_{2G}^{(T)}\right)  &=& I\left(X_{2G}^{(T)};Y_{2G}^{(T)}\right)
\end{eqnarray*}
implying that 
\begin{eqnarray*}
R_{11}+R_{12}+R_{21}+R_{22} &\leq&  \frac{1}{2} \log \left( 1 + \frac{H_{11}^2 P_1}{1 + H_{12}^2 P_2} \right) +
\frac{1}{2} \log \left( 1 + \frac{H_{22}^2 P_2}{1 + H_{21}^2 P_1} \right)
\end{eqnarray*}
Also as shown in \cite{Sreekanth_Veeravalli},  (\ref{eq:worstnoisecond}), (\ref{eq:worstnoisecond2}), (\ref{eq:noisyintcond1}), (\ref{eq:noisyintcond2}) can be combined as 
\begin{equation}
\left| \frac{H_{12}}{H_{22}} \left( 1 + H_{21}^2 P_1 \right) \right| +
\left| \frac{H_{21}}{H_{11}} \left( 1 + H_{12}^2 P_2 \right) \right| \leq 1,
\end{equation}
Equations (\ref{eqn:prt2cpcty}), (\ref{eqn:prt2cndt}) can be derived similarly. This completes the proof.
\end{proof}

\section{Conclusions}
\label{sec:conclusion}
We found the generalized degrees of freedom(GDOF) of the symmetric two-user Gaussian $X$ channel. To find the GDOF of the $X$ channel, we first found the sum capacity of a deterministic $X$ channel and extended insights gained from the deterministic case to obtain the GDOF of the Gaussian channel.  In the process, we found an outerbound for the sum capacity of the two-user Gaussian $X$ channel that coincides with the bound on the sum capacity of the two-user interference channel derived by Etkin, Tse and Wang in reference \cite{Etkin_Tse_Wang}. The implication of the bound is that, for certain regimes, the performance of the $X$ channel  is identical to the performance of the two-user interference channel from a GDOF perspective. However, for other regimes, we showed that the $X$ channel outperforms the interference channel through an interference alignment based achievable scheme. Our result therefore characterizes the benefits obtained from interference alignment from a GDOF perspective. While our results characterize the GDOF of the $X$ channel in the symmetric setting, an interesting and important area of future work lies in extending the study of the general setting which is not symmetric. In particular, there lies open the question of whether new outerbounds are required to characterize the GDOF in the asymmetric case, or whether the current bounds are tight. In the Gaussian multiple access, broadcast and two-user interference networks, the capacity of the appropriate deterministic channel is within a constant number of bits of the corresponding Gaussian channel. It is an imporant open question whether the solution to the deterministic $X$ channel provided in this work leads to useful approximations of the capacity of the Gaussian $X$ channel.

As a by-product of the main result, we also extended bounds derived for the interference channel in \cite{MK_int,Shang_Kramer_Chen,Sreekanth_Veeravalli} to the two-user $X$ channel. The bound implies that, for certain class of channel coefficients, it is capacity optimal in the two-user $X$ channel to set two messages to null so that it forms an interference channel, encode both non-null messages using Gaussian codebooks and decode at both receivers by treating interference as noise. Therefore, interestingly, for a class of channel coefficients, certain messages do not contribute to the sum capacity in the $X$ channel. An interesting open question related to this result is whether there exist channel coefficients in the two-user and/or larger $X$ networks, where setting other sets of messages to null is sum-rate optimal.

\newpage
\appendices

\section{Proof of Lemma \ref{lemma:sym}}
\label{sec:ProofLemmaSym}
The symmetric interference channel maybe represented as 
\begin{eqnarray*} 
Y_1(\tau) &=& \rho X_1(\tau) + \rho^{\alpha} X_2(\tau) + Z_1(\tau)\\
Y_2(\tau) &=& \rho^{\alpha} X_1(\tau) + \rho X_2(\tau) + Z_2(\tau)
\end{eqnarray*}

Now, by simply switching the two receivers in the $X$ channel, the input-output relations maybe alternately described as 
\begin{eqnarray*} 
Y_1^{'}(\tau) &=& \rho^{'} X_1(\tau) + (\rho^{'})^{\alpha^{'}} X_2(\tau) + Z_1^{'}(\tau) \label{xcheq1} \\
Y_2^{'}(\tau) &=&  (\rho^{'})^{\alpha^{'}} X_1(\tau) + \rho^{'} X_2(\tau) +Z_2^{'}(\tau) \label{xcheq2} 
\end{eqnarray*}
where
\begin{eqnarray*}
Y_1^{'} = Y_{2}&,& Y_{2}^{'} = Y_{1}\\ 
Z_{1}^{'} = Z_{2}&,& Z_{2}^{'} = Z_{1}\\
\rho^{'}=\rho^{\alpha}&,& \alpha^{'} = \frac{1}{\alpha}
\end{eqnarray*}
Note that the capacity of the $X$ channel described in equations (\ref{xcheq1}),(\ref{xcheq2}) is $C_{\Sigma}(\rho^{'},\alpha^{'})$. Further more, since simply switching the receivers of the original $X$ channel does not alter the sum capacity, we can write 
$$ C_{\Sigma}(\rho^{'},\alpha^{'}) = C_{\Sigma}(\rho,\alpha)$$
\begin{eqnarray*}
\Rightarrow \lim_{\rho \to \infty} \frac{C_{\Sigma}(\rho^{'},\alpha^{'})}{\frac{1}{2}\log(\rho)} &=& \lim_{\rho \to \infty} \frac{C_{\Sigma}(\rho,\alpha)}{\frac{1}{2}\log(\rho)}\\
\Rightarrow \lim_{\bar{\rho} \to \infty} \alpha \frac{C_{\Sigma}(\rho^{'},\alpha^{'})}{\frac{1}{2}\log(\rho^{'})} &=& d(\alpha)\\
\Rightarrow \alpha d(\alpha^{'}) &=& d(\alpha)\\
\Rightarrow \alpha d(\frac{1}{\alpha}) &=& d(\alpha)
\end{eqnarray*}

\section{Proof of Lemma \ref{lem:extdecomp}}
\label{sec:ProofLemmaDecomp}
	We start with a proof of part 1 of the lemma. 
Before going into the detail, we want to point that Lemma \ref{lem:extdecomp} is for $(n_c, n_d)$ such that $\frac{3}{4} \leq \frac{n_c}{n_d} < 1$. For part 1, it should also be noted that $\frac{n_c}{3}$ is a positive integer. To simply notation usage, let
\begin{eqnarray}
\mathbf{H} = \mathbf{S}^{n_d - n_c}.
\end{eqnarray}
$\mathcal{F}_2^{n_d}$ can be expressed as the following
\begin{eqnarray}
\label{eqn:Fnd}
\mathcal{F}_2^{n_d}
&=&
\textrm{span} \left( \mathbf{e}_1,\mathbf{e}_2,\mathbf{e}_3,\ldots,\mathbf{e}_{n_d-1},
\mathbf{e}_{n_d} \right) \\
& \stackrel{(a)}{=} &
\label{eqn:recur}
\textrm{span}
\left(
\mathbf{e}_1,\mathbf{e}_2,\ldots,\mathbf{e}_{n_d-n_c},
\mathbf{H} \mathbf{e}_1,
\mathbf{H} \mathbf{e}_2, \ldots,
\mathbf{H} \mathbf{e}_{n_d-n_c}, \ldots,
\mathbf{H}^n \mathbf{e}_k
\right) \\
& \stackrel{(b)}{=} &
\textrm{span}
\left(
\mathbf{e}_1, \mathbf{H} \mathbf{e}_1,
\mathbf{H}^2 \mathbf{e}_1,\ldots,
\mathbf{H}^{\left\lfloor \frac{n_d - 1}{n_d - n_c} \right\rfloor} \mathbf{e}_1
\right)
\oplus
\textrm{span}
\left(
\mathbf{e}_2, \mathbf{H} \mathbf{e}_2,
\mathbf{H}^2 \mathbf{e}_2, \ldots,
\mathbf{H}^{\left\lfloor \frac{n_d - 2}{n_d - n_c} \right\rfloor} \mathbf{e}_2
\right) \notag \\
& &
\oplus \cdots \oplus
\textrm{span}
\left(
\label{eqn:CyclcDec}
\mathbf{e}_{n_d - n_c}, \mathbf{H} \mathbf{e}_{n_d - n_c},
\mathbf{H}^2 \mathbf{e}_{n_d - n_c},\ldots,
\mathbf{H}^{\left\lfloor \frac{n_d - (n_d - n_c)}{n_d - n_c} \right\rfloor} \mathbf{e}_{n_d - n_c}
\right)
\end{eqnarray}
where $\mathbf{e}_i$ is the $i_{\textrm{th}}$ column vector of $\mathbf{I}_{n_d}$, an identity matrix in $\mathcal{F}_2^{n_d \times n_d}$, and $\oplus$ is the direct sum operator for subspaces. Note that in step (a), we recursively use the property that when $0 < i \leq n_c$, we have
\begin{eqnarray}
\mathbf{e}_{n_d - n_c + i} = \mathbf{S}^{n_d - n_c} \mathbf{e}_i.
\end{eqnarray}
The $(n,k) \in \mathbb{N}^2$ in (\ref{eqn:recur}) satisfies $n(n_d-n_c)+k=n_d$ and $1 \leq k \leq n_d - n_c$, and can be uniquely determined. In step (b), we reorganize the basis and divide the basis into several subsets. An example of the case that $(n_c, n_d) = (10, 13)$ is illustrated in Fig. \ref{fig:CyclcDcmp}.

\begin{figure}[!tb]
\begin{center}
\includegraphics[angle = 90, width=400pt, trim=80 0 30 0, clip]{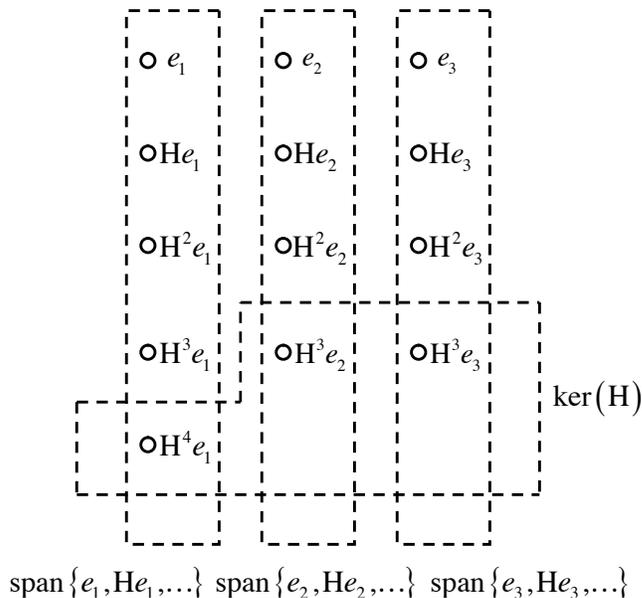}
\caption{A pictorial representation of the cyclic decomposition of $\mathcal{F}_2^{n_d}$ with $(n_c, n_d) = (10,13)$.}\label{fig:CyclcDcmp}
\end{center}
\end{figure}

\begin{figure}[!tb]
\begin{center}
\includegraphics[angle = 90, width=400pt, trim=80 0 60 180, clip]{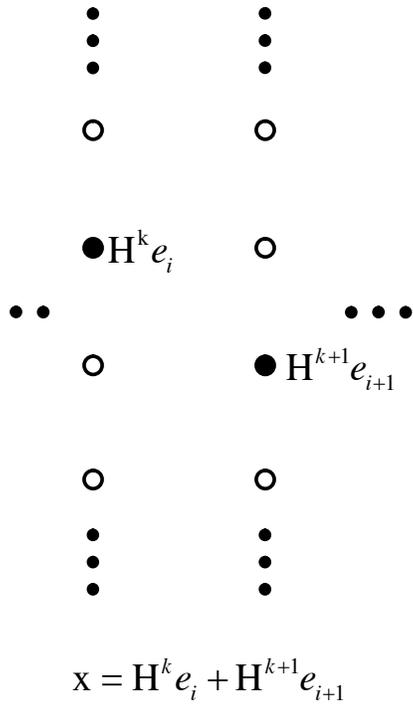}
\caption{The pictorial representation of the linear combination of two vectors: $\mathbf{H}^k \mathbf{e}_i$ and $\mathbf{H}^{k+1} \mathbf{e}_{i+1}$.}\label{fig:LnrCmb}
\end{center}
\end{figure}

\begin{figure}[!tb]
\begin{center}
\includegraphics[angle=90, width=400pt, trim=250 0 50 0, clip]{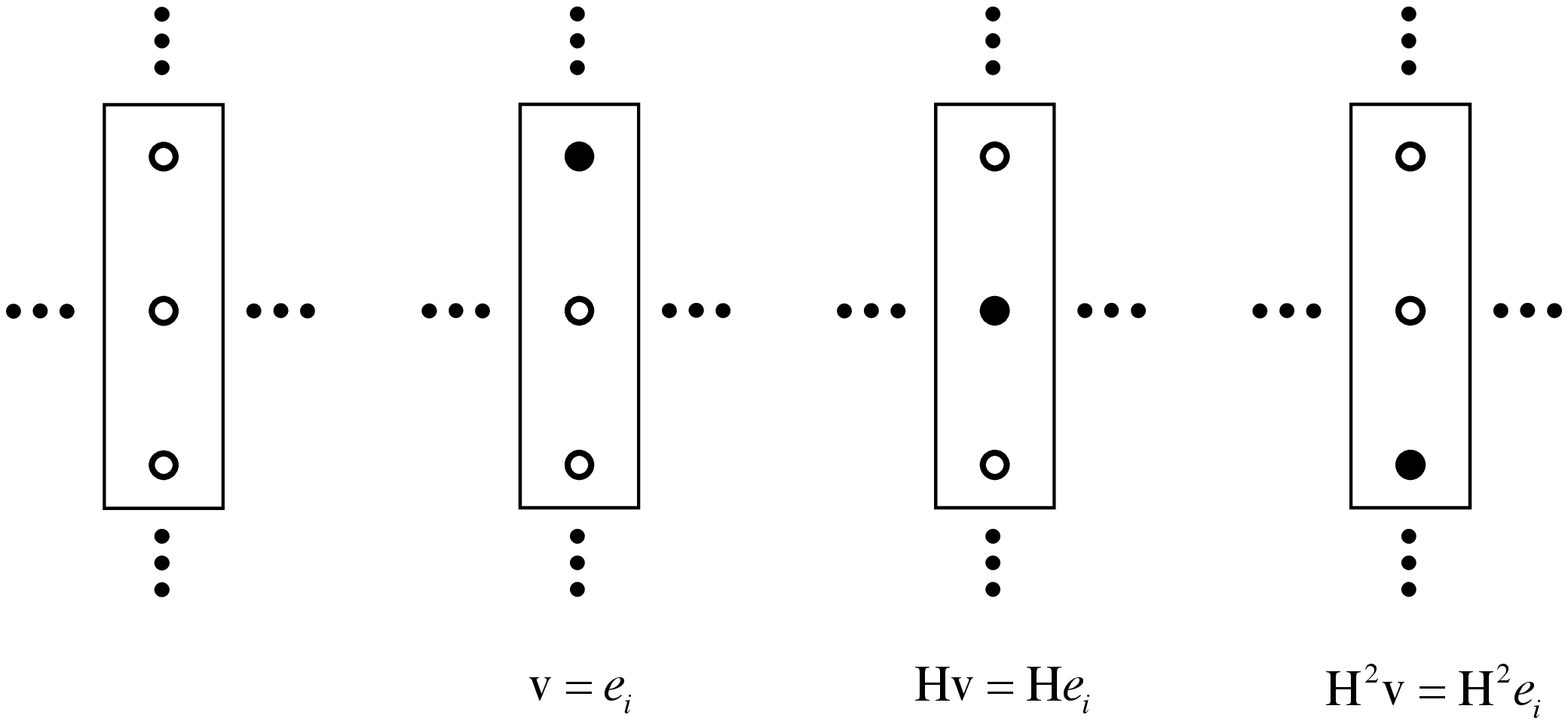}
\caption{Box 1 and the vectors associated with it.}\label{fig:bx1}
\end{center}
\end{figure}

\begin{figure}[!tb]
\begin{center}
\includegraphics[angle=90, width=400pt, trim=190 0 65 0, clip]{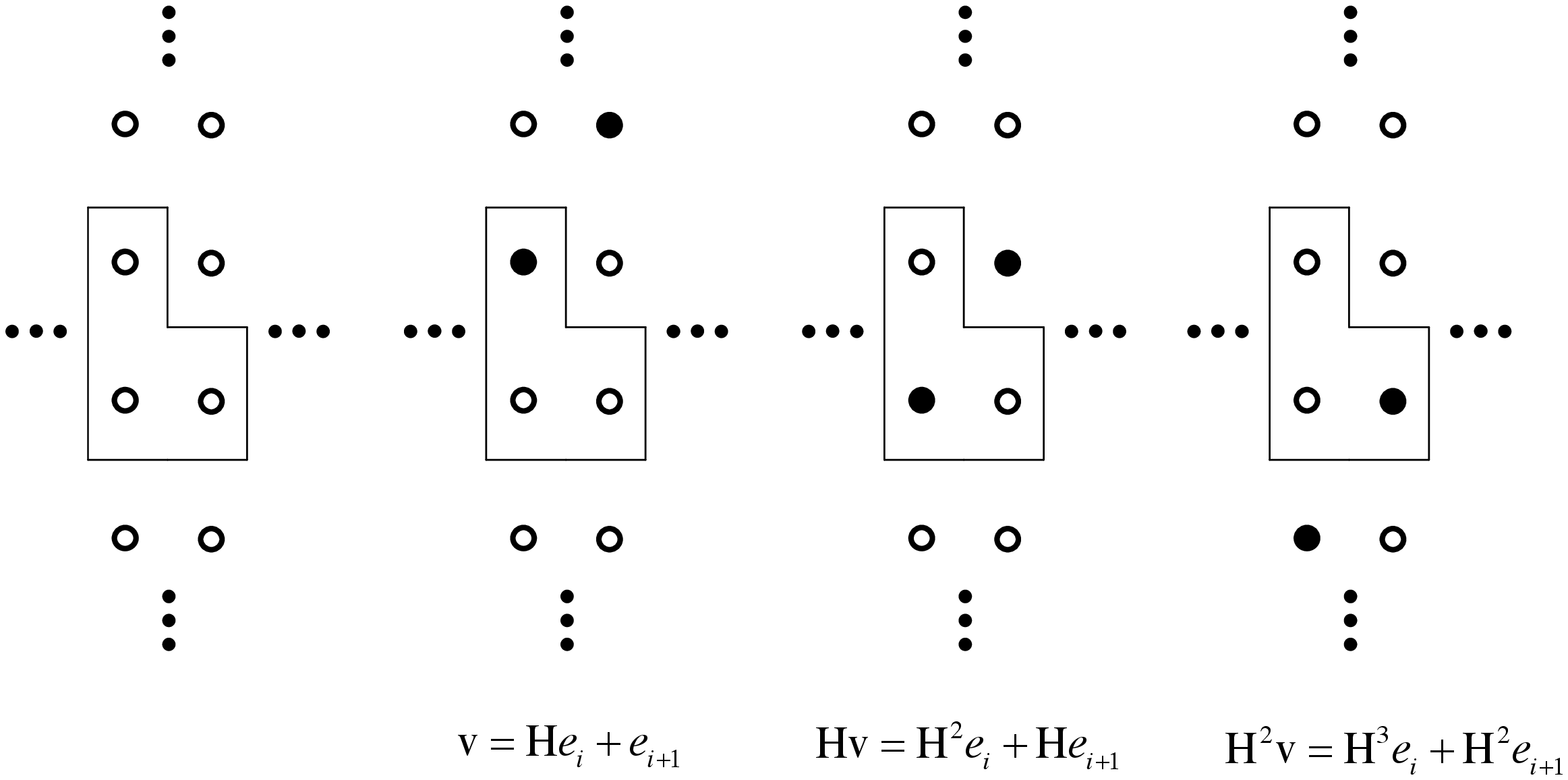}
\caption{Box 2 and the vectors associated with it.}\label{fig:bx2}
\end{center}
\end{figure}

\begin{figure}[!tb]
\begin{center}
\includegraphics[angle=90, width=400pt, trim=190 0 65 0, clip]{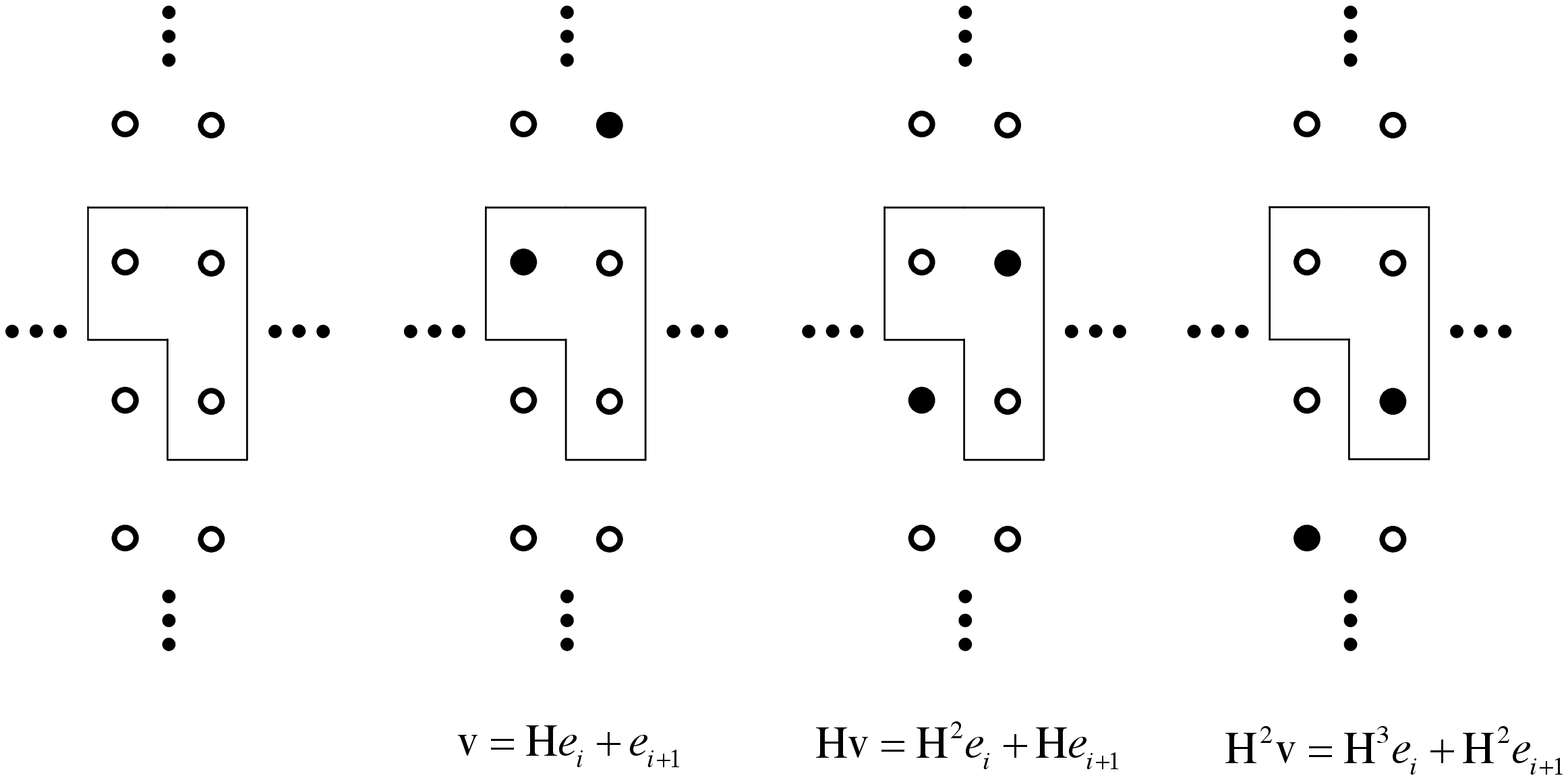}
\caption{Box 3 and the vectors associated with it.}\label{fig:bx3}
\end{center}
\end{figure}

\begin{figure}[!tb]
\begin{center}
\includegraphics[angle=90, width=400pt, trim=130 0 65 0, clip]{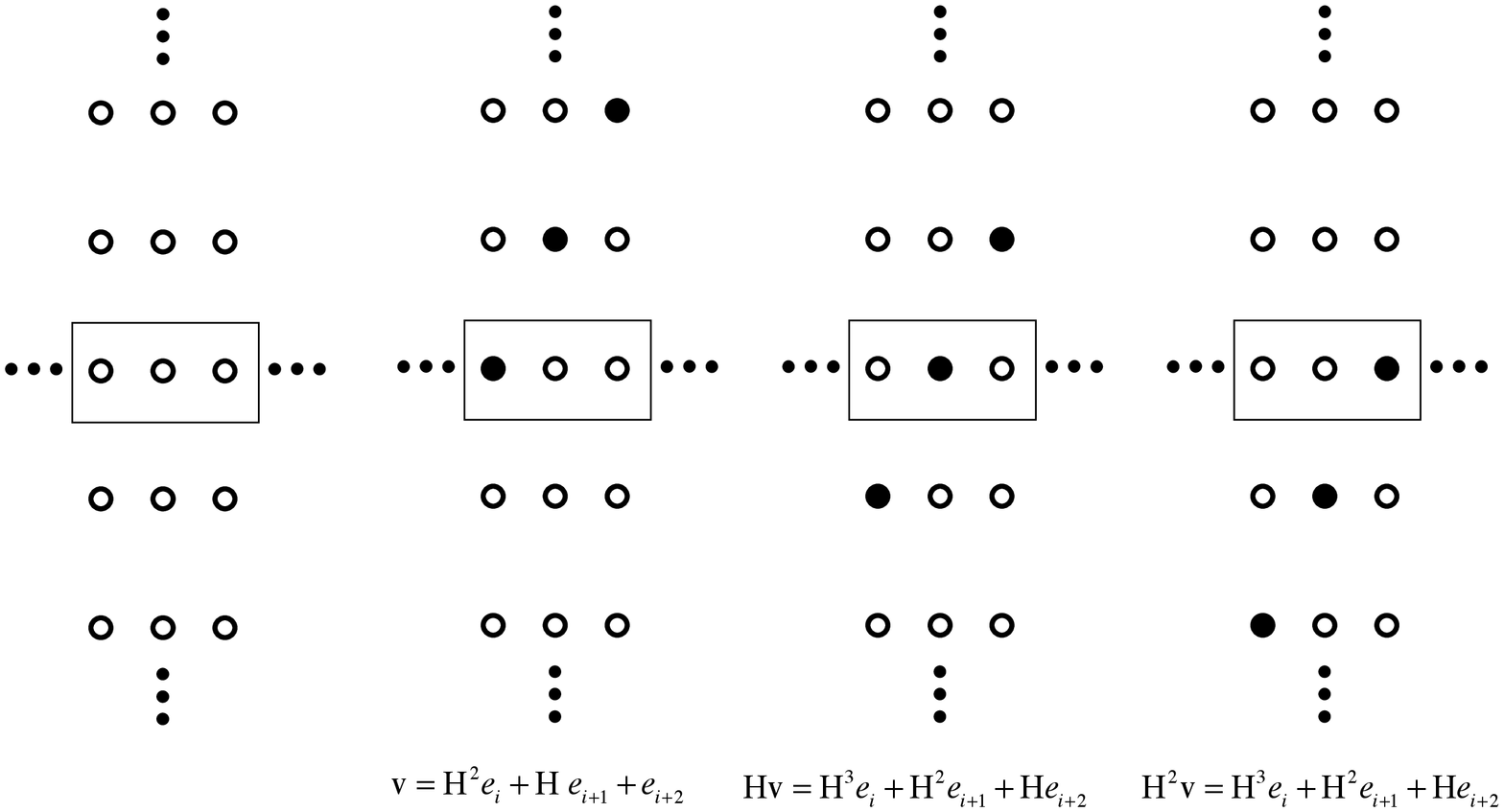}
\caption{Box 4 and the vectors associated with it.}\label{fig:bx4}
\end{center}
\end{figure}
\begin{figure}[!tb]
\begin{center}
\includegraphics[angle=90, width=400pt, trim=190 0 65 0, clip]{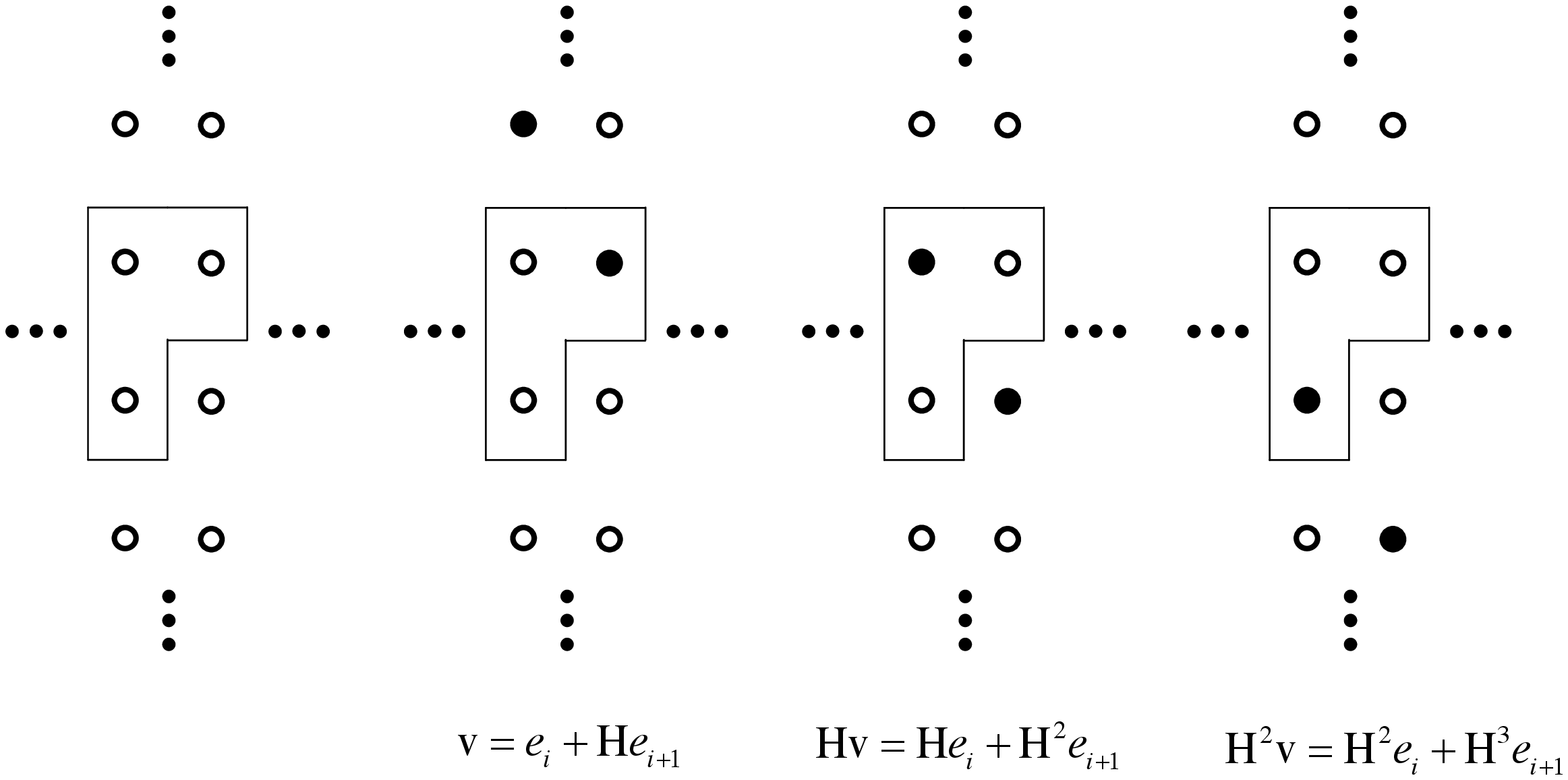}
\caption{Box 5 and the vectors associated with it.}\label{fig:bx5}
\end{center}
\end{figure}
\begin{figure}[!tb]
\begin{center}
\includegraphics[angle=90, width=400pt, trim=50 0 80 0, clip]{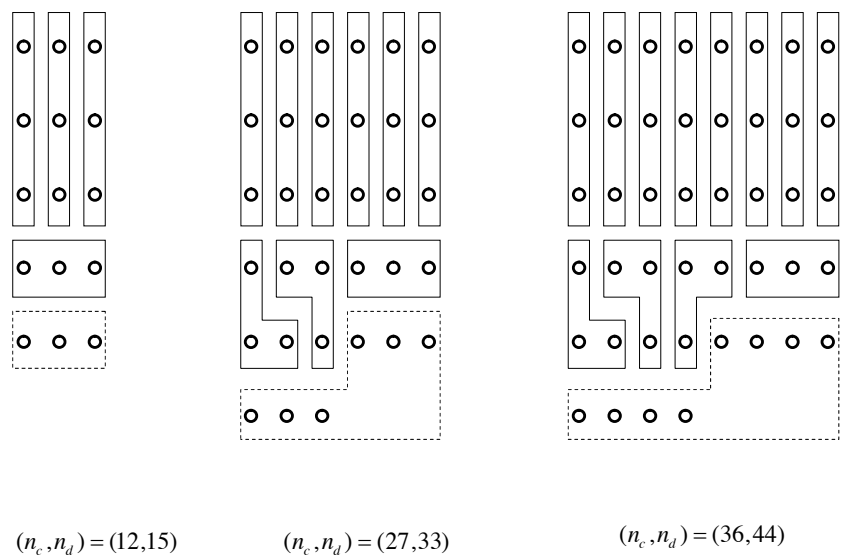}
\caption{Illustration of the algorithm.}\label{fig:GrpCyclcDcmp}
\end{center}
\end{figure}

\begin{figure}[!tb]
\begin{center}
\includegraphics[angle=90, width=400pt, trim=150 0 65 0, clip]{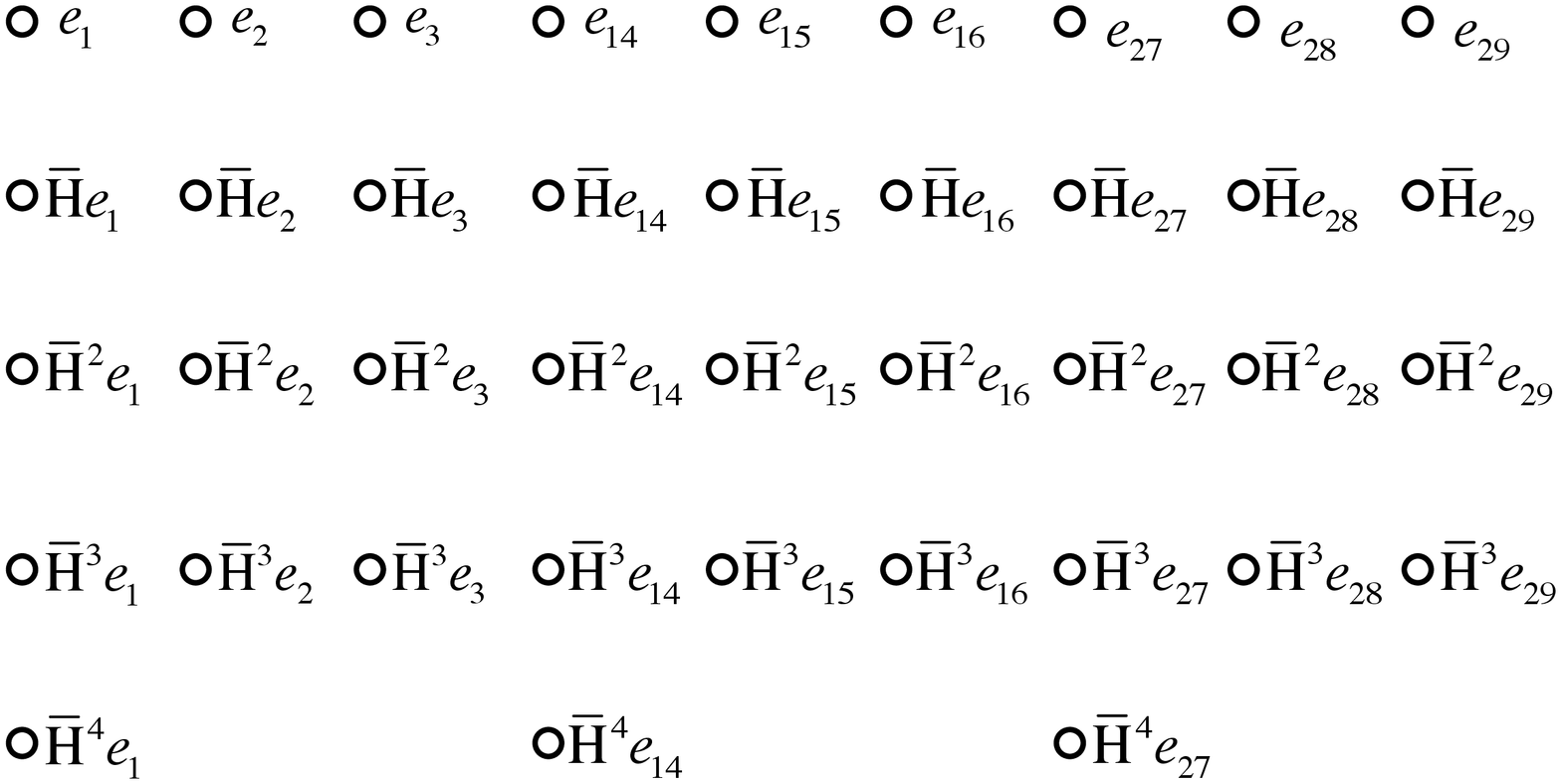}
\caption{A pictorial representation of the cyclic decomposition of $\mathcal{F}_2^{n_d}$ representing the signal space of the 3-symbol extension of the case that $(n_c, n_d) = (10, 13)$.}\label{fig:ExtdCyclcDcmp}
\end{center}
\end{figure}

\begin{figure}[!tb]
\begin{center}
\includegraphics[angle=90, width=400pt, trim=150 0 80 0, clip]{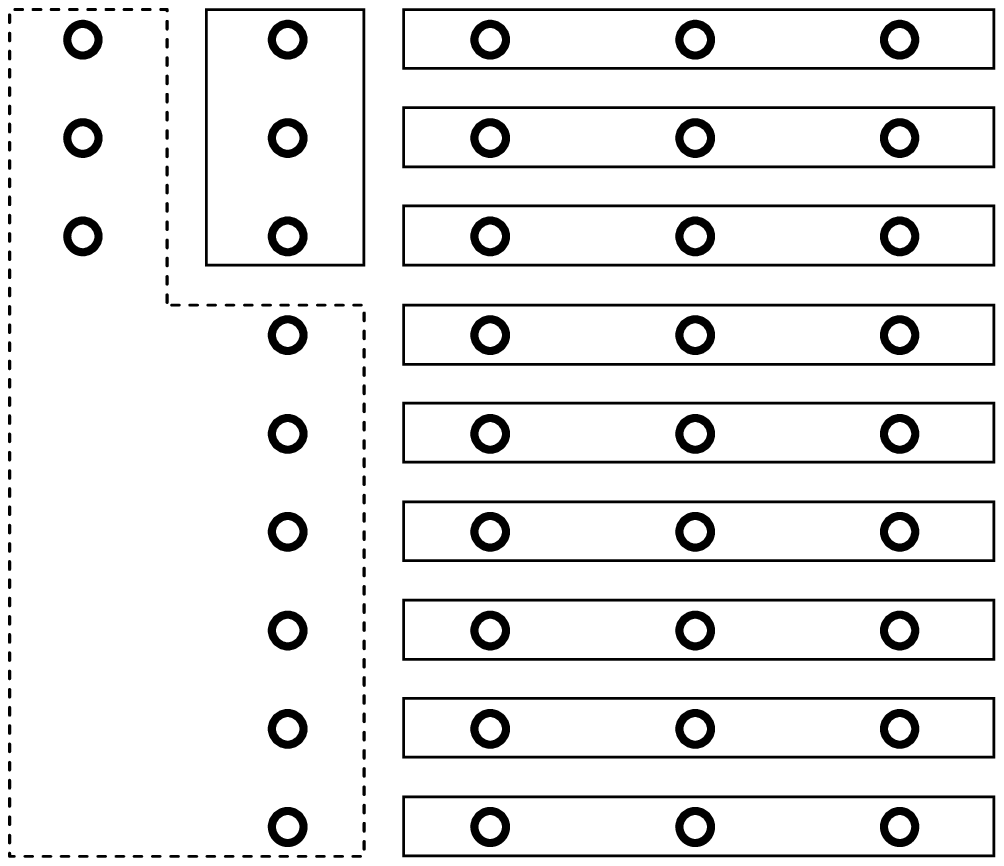}
\caption{A pictorial representation of the cyclic decomposition of $\mathcal{F}_2^{n_d}$ representing the signal space of the 3-symbol extension of the case that $(n_c, n_d) = (10, 13)$ after reordering and grouping.}\label{fig:GrpExtdCyclicDcmp}
\end{center}
\end{figure}

Conceptually, (\ref{eqn:CyclcDec}) decomposes $\mathcal{F}_2^{n_d}$ into several disjoint subspaces. This decomposition is called cyclic decomposition in the content of linear algebra \cite{linear_algebra}. There are some interesting properties for this decomposition. First, multiplying $\mathbf{H}$ to any vector lying in the subspace spanned by $\{\mathbf{e}_i, \mathbf{H} \mathbf{e}_i, \mathbf{H}^2 \mathbf{e}_i,\ldots \}$ results in another vector lying in the same subspace. Thus, $\textrm{span} \{\mathbf{e}_i, \mathbf{H} \mathbf{e}_i, \mathbf{H}^2 \mathbf{e}_i,\ldots \}$ is called an $\mathbf{H}$-invariant subspace of $\mathcal{F}_2^{n_d}$. Second, the total number of the $\mathbf{H}$-invariant subspaces for $\mathcal{F}_2^{n_d}$ is equal to the number of the dimensions of $\ker( \mathbf{H} )$ (Theorem 8.2.19 in \cite{linear_algebra}). Third, $\mathcal{F}_2^{n_d}$ can be expressed as the direct sum of all $\mathbf{H}$-invariant subspaces. Due to the specific structure of $\mathbf{S}^{n_d - n_c}$, the difference between the number of dimensions of any two $\mathbf{H}$-invariant subspaces is less than or equal to one.

After obtaining a new basis for $\mathcal{F}_2^{n_d}$, we can express any vector in $\mathcal{F}_2^{n_d}$ as a linear combination of the vectors in the basis. We use Fig. \ref{fig:LnrCmb} to illustrate our usage of notations. We then introduce five different boxes illustrated in Fig. \ref{fig:bx1} to \ref{fig:bx5}. Each box contains three circles and represents a set of three vectors $\{ \mathbf{v}, \mathbf{H} \mathbf{v}, \mathbf{H}^2 \mathbf{v} \}$ for some $\mathbf{v}$ in $\mathcal{F}_2^{n_d}$. Note that these vectors are linearly independent.

We then use the following algorithm, including five steps, to decompose a plot representing a cyclic decomposition of $\mathcal{F}_2^{n_d}$ into a set of circles representing the basis of $\ker(\mathbf{H})$ and a collection of the boxes shown in Fig. \ref{fig:bx1} to \ref{fig:bx5}.

\textit{STEP 1:} Collect the circles located in the bottom of each column. This gives us the set of circles representing the basis for $\ker(\mathbf{H})$.

\textit{STEP 2:} Starting from the top of each column, put as many boxes shown in Fig. \ref{fig:bx1} as possible in each column. Note that the condition $\frac{n_c}{n_d} \geq \frac{3}{4}$ ensures that at least one such box can be put in each column. After this step, each column has at most two unassigned circles.

\textit{STEP 3:} For the remaining unassigned circles, starting from the left-most place and then gradually moving to the right, alternatively put as many boxes shown in Fig. \ref{fig:bx2} and boxes shown in Fig. \ref{fig:bx3} as possible.

\textit{STEP 4:} For the remaining unassigned circles, starting from the right-most place and then gradually moving to the left, put as many boxes shown in Fig. \ref{fig:bx4} as possible.

\textit{STEP 5:} After steps 1 to 4, if there are still some unassigned circles, they would have the exactly same shape as the box shown in Fig. \ref{fig:bx5} Thus, we could use the box to group the remaining circles. This is the end of all steps.

The algorithm is illustrated in Fig. \ref{fig:GrpCyclcDcmp}. Note that at the end of steps 2 to 4, if there are no unassigned circle, the algorithm is terminated immediately. After step 1, there are $n_c$ unassigned circles. Because each box contains three circles, we need a total number of $\frac{n_c}{3}$ boxes to assign all the circles. Now we are ready to find the column vectors of $\mathbf{V}$.

Let $\mathbf{v}_i$ be the first vector represented by a specific box for $i=1,2,\ldots,\frac{n_c}{3}$. Let $\mathbf{V} \in \mathcal{F}_2^{n_d \times \frac{n_c}{3}}$ be constructed as
\begin{eqnarray}
\mathbf{V} =
\left[
\begin{array}{cccc}
\mathbf{v}_1 & \mathbf{v}_2 & \cdots & \mathbf{v}_{\frac{n_c}{3}}
\end{array}
\right].
\end{eqnarray}
Now consider the following matrix
\begin{eqnarray}
&&
\left[
\begin{array}{ccc}
\mathbf{V} & \mathbf{HV} & \mathbf{H}^2 \mathbf{V}
\end{array}
\right] \notag \\
&=&
\left[
\begin{array}{c|c|c}
\begin{array}{cccc}
\mathbf{v}_1 & \mathbf{v}_2 & \cdots & \mathbf{v}_{\frac{n_c}{3}}
\end{array}
&
\begin{array}{cccc}
\mathbf{H} \mathbf{v}_1 & \mathbf{H} \mathbf{v}_2 & \cdots & \mathbf{H} \mathbf{v}_{\frac{n_c}{3}}
\end{array}
&
\begin{array}{cccc}
\mathbf{H}^2 \mathbf{v}_1 & \mathbf{H}^2 \mathbf{v}_2 & \cdots & \mathbf{H}^2 \mathbf{v}_{\frac{n_c}{3}}
\end{array}
\end{array}
\right]
\end{eqnarray}
Using the fact that$\{ \mathbf{v}_i, \mathbf{H} \mathbf{v}_i, \mathbf{H}^2 \mathbf{v}_i \}$ are linearly independent for $i \in \{ 1,2,\ldots,\frac{n_c}{3} \}$ and the fact that vectors represented by a specific box can not be written as a linear combination of vectors represented by the other boxes, we have the result that all column vectors of $[ \begin{array}{ccc} \mathbf{V} & \mathbf{HV} & \mathbf{H}^2 \mathbf{V} \end{array} ]$ are linearly independent. Therefore, we have
\begin{eqnarray}
\textrm{rank}\left( \left[
\begin{array}{ccc}
\mathbf{V} & \mathbf{HV} & \mathbf{H}^2 \mathbf{V}
\end{array}
\right] \right) = n_c
\end{eqnarray}
and
\begin{eqnarray}
\mathcal{F}_2^{n_d} = \ker(\mathbf{H}) \oplus \textrm{col}(\mathbf{V}) \oplus \textrm{col}(\mathbf{HV}) \oplus \textrm{col}(\mathbf{H}^2 \mathbf{V}).
\end{eqnarray}
This concludes the proof of part 1 of the lemma. We now proceed to part 2.


The proof follows the similar steps with those for part 1, but we need some extra arrangements to deal with channel extension.

$\mathcal{F}_2^{3n_d}$ can be expressed as the following
\begin{eqnarray}
\label{eqn:F3nd}
\mathcal{F}_2^{3n_d}
&=&
\textrm{span} \left( \mathbf{e}_1,\mathbf{e}_2,\mathbf{e}_3,\ldots,\mathbf{e}_{n_d-1},
\mathbf{e}_{n_d} \right)
\oplus
\textrm{span} \left( \mathbf{e}_{n_d+1},\mathbf{e}_{n_d+2},\mathbf{e}_{n_d+3},\ldots,\mathbf{e}_{2n_d-1},
\mathbf{e}_{2n_d} \right) \notag \\
&&
\oplus
\textrm{span} \left( \mathbf{e}_{2n_d+1},\mathbf{e}_{2n_d+2},\mathbf{e}_{2n_d+3},\ldots,\mathbf{e}_{3n_d-1},
\mathbf{e}_{3n_d} \right) \\
& \stackrel{(a)}{=} &
\label{eqn:recurext}
\textrm{span}
\left(
\mathbf{e}_1,\ldots,\mathbf{e}_{n_d-n_c},
\mathbf{\bar{H}} \mathbf{e}_1, \ldots,
\mathbf{\bar{H}} \mathbf{e}_{n_d-n_c}, \ldots,
\mathbf{\bar{H}}^n \mathbf{e}_k
\right) \notag \\
&&
\oplus
\textrm{span}
\left(
\mathbf{e}_{n_d+1}, \ldots,\mathbf{e}_{2n_d-n_c},
\mathbf{\bar{H}} \mathbf{e}_{n_d+1}, \ldots,
\mathbf{\bar{H}} \mathbf{e}_{2n_d-n_c}, \ldots,
\mathbf{\bar{H}}^n \mathbf{e}_{n_d+k}
\right) \notag \\
&&
\oplus
\textrm{span}
\left(
\mathbf{e}_{2n_d+1},\ldots,\mathbf{e}_{3n_d-n_c},
\mathbf{\bar{H}} \mathbf{e}_{2n_d+1}, \ldots,
\mathbf{\bar{H}} \mathbf{e}_{3n_d-n_c}, \ldots,
\mathbf{\bar{H}}^n \mathbf{e}_{2n_d+k}
\right)
\end{eqnarray}
where the notation usage is similar with those used in previous section. The $(n,k) \in \mathbb{N}^2$ in (\ref{eqn:recurext}) satisfies $n(n_d-n_c)+k=n_d$ and $1 \leq k \leq n_d - n_c$, and can be uniquely decided. Note that there are three disjoint subspaces in (\ref{eqn:recurext}), and we can apply the similar decomposition used in (\ref{eqn:CyclcDec}) to decompose each subspace. An example of the 3-symbol extension of the case that $(n_c, n_d) = (10, 13)$ is illustrated in Fig. \ref{fig:ExtdCyclcDcmp}.

One can easily observe that a part of the plot is duplicated twice to form the whole plot, and the part that is duplicated has the same structure with those in Appendix \ref{sec:ProofLemmaDecomp}. Since each column represents a basis for an $\mathbf{\bar{H}}$-invariant subspace, we can simply reorder the columns to let the new plot have the same structure with those in Appendix \ref{sec:ProofLemmaDecomp}. The idea is illustrated in Fig. \ref{fig:GrpExtdCyclicDcmp}. Also note that the number of circles not at the bottom of each column is $3n_c$ which is a multiple of three. Thus we can use the algorithm introduced in Appendix \ref{sec:ProofLemmaDecomp} to decompose the plot and obtain a $\mathbf{V} \in \mathcal{F}_2^{3n_d \times n_c}$ such that
\begin{eqnarray}
\textrm{rank}\left( \left[
\begin{array}{ccc}
\mathbf{V} & \mathbf{\bar{H}} \mathbf{V} & \mathbf{\bar{H}}^2 \mathbf{V}
\end{array}
\right] \right) = 3n_c
\end{eqnarray}
and
\begin{eqnarray}
\mathcal{F}_2^{3n_d} = \ker(\mathbf{\bar{H}}) \oplus \textrm{col}(\mathbf{V}) \oplus \textrm{col}(\mathbf{\bar{H}} \mathbf{V}) \oplus \textrm{col}(\mathbf{\bar{H}}^2 \mathbf{V}).
\end{eqnarray}
This concludes our proof.

\section{Proof of Lemma \ref{lem:knownob}}
\label{append:zbound}
We intend to prove (\ref{z-bound1})-(\ref{z-bound4}) here. We only show (\ref{z-bound1}). All the other bounds follow by symmetry. Since we intend to bound $R_{11}+R_{22}+R_{12}$, we set $W_{21}=\phi$ and show 
\begin{eqnarray*}R_{11}+R_{22}+R_{12} &\leq& \log\big(1+H_{11}^2 P_1+ H_{12}^2 P_2\big)+ \log\big(1+\frac{H_{22}^2 P_2}{1+H_{12}^2 P_2}\big) \end{eqnarray*}
Note that setting $W_{21}=\phi$ does not affect the converse argument since it does not reduce the rates of the other messages. Now, we let a genie provide 
$Y_{1}^{(T)},W_{11},W_{12}$ to receiver $2$. Now, using Fano's inequality, we can bound the sum-rate $R_{11}+R_{12}+R_{22}$ as follows
\begin{eqnarray}
T R_{22} + T R_{11} + T R_{12} - T\epsilon &\leq& I(Y_1^{(T)}; W_{11}, W_{12})+ I(Y_{2}^{(T)}, Y_{1}^{(T)}, W_{11},W_{12}; W_{22})\\
&\leq& I\left(Y_1^{(T)}; W_{11}, W_{12}\right)+I\left(Y_{2}^{(T)}, Y_{1}^{(T)}, ; W_{22}|W_{11},W_{12}\right)+I\left(W_{11},W_{12}; W_{22}\right)\label{zbound:ineq1}\\
&\leq& I\left(Y_1^{(T)}; W_{11}, W_{12}\right)+I\left(Y_{2}^{(T)}, Y_{1}^{(T)}, ; W_{22}|W_{11},W_{12}\right)\label{zbound:ineq2} \\
&\leq& h\left(Y_1^{(T)}\right) - h\left(Y_1^{(T)}| W_{11}, W_{12}\right)+ h\left(Y_{2}^{(T)}, Y_{1}^{(T)}| W_{11},W_{12} \right) \nonumber \\ && - h\left(Y_{2}^{(T)}, Y_{1}^{(T)}| W_{22}, W_{11},W_{12}\right) \label{zbound:ineq3}\\
&\leq& h\left(Y_1^{(T)}\right) + h\left(Y_{2}^{(T)}| Y_{1}^{(T)} W_{11},W_{12} \right) - h\left(Y_{2}^{(T)}, Y_{1}^{(T)}| W_{22}, W_{11},W_{12}\right) \label{zbound:ineq4}\\
&\leq& h\left(Y_1^{(T)}\right) +h\left(Y_{2}^{(T)} | Y_{1}^{(T)}, W_{11},W_{12},X_1^{(T)}\right)  \nonumber \\&&- h\left(Y_{2}^{(T)}, Y_{1}^{(T)}| W_{22}, W_{11},W_{12},X_2^{(T)}, X_1^{(T)}\right) \label{zbound:ineq5}\\
&\leq& h\left(Y_1^{(T)}\right)+h\left(S_{22}^{(T)}| S_{12}^{(T)}, W_{11},W_{12},X_1^{(T)}\right) \nonumber \\ &&- h\left(Z_{2}^{(T)}, Z_{1}^{(T)}| W_{22}, W_{11},W_{12},X_2^{(T)}, X_1^{(T)}\right) \label{zbound:ineq6}\\
&\leq& h\left(Y_{1}^{(T)}\right)+h\left(S_{22}^{(T)}|S_{12}^{(T)}\right) - h\left(Z_{2}^{(T)}, Z_{1}^{(T)}\right) \label{zbound:ineq7}\\
&\leq& \displaystyle\sum_{\tau=1}^{T} h\left(Y_1(\tau)\right)+\sum_{\tau=1}^{T} h\left(S_{22}(\tau)|S_{12}(\tau)\right) - h\left(Z_{2}^{(T)}, Z_1^{(T)}\right)\label{zbound:ineq8}\\
&\leq& T \log\left(1+H_{11}^2 P_1+H_{12}^2 P_2\right) + T\log\left(1+\frac{H_{22}^2 P_2}{1+H_{12}^2 P_2}\right)
\end{eqnarray}
where, in (\ref{zbound:ineq1}), the second term is zero since all messages in the system are independent of each other. The second term in (\ref{zbound:ineq4}) is obtained by combining the second and third summands of (\ref{zbound:ineq3}) using the chain rule. In the first summand on the right hand side in (\ref{zbound:ineq5}), we have used the fact that given $W_{11}$, $X_1$ is known at receiver $2$, since $W_{21}=\phi$. In the second term in (\ref{zbound:ineq5}), we have used the fact that conditioning on $X_1^{(T)}$ and $X_2^{(T)}$ does not reduce entropy. In (\ref{zbound:ineq6}), we have cancelled the effect of $X_1^{(T)}$ from $Y_{1}^{(T)},Y_{2}^{(T)}$. 	 In (\ref{zbound:ineq7}), we have used the fact that conditioning does not reduce entropy, and in the final term, we use the independence of the noise terms w.r.t the inputs and messages in the systems. In the final two steps, we have used the convexity of mutual information, and the fact that that circularly symmetric Gaussian variables maximize differential and conditional entropy under a covariance constraint. The bounds in (\ref{z-bound2})-(\ref{z-bound4}) can be shown by applying similar arguments as above to the appropriate $Z$ channel.
This completes the proof.

\newpage
\bibliography{Thesis}
\end{document}